\documentclass[acmsmall,screen]{acmart}

\setcopyright{rightsretained}
\acmPrice{}
\acmDOI{10.1145/3290344}
\acmYear{2019}
\copyrightyear{2019}
\acmJournal{PACMPL}
\acmVolume{3}
\acmNumber{POPL}
\acmArticle{31}
\acmMonth{1}

\usepackage{booktabs}   %% For formal tables:
\usepackage{subcaption} %% For complex figures with subfigures/subcaptions
\usepackage{qcircuit}
\usepackage{graphicx}
\usepackage{mathtools}
\usepackage{proof}
\usepackage{enumitem}

\newtheorem{theorem}{Theorem} %[section]

\newtheorem{proposition}[theorem]{Proposition}

\newtheorem{example}{Example}
\newtheorem{remark}[theorem]{Remark}
\newtheorem{definition}[theorem]{Definition}

\numberwithin{equation}{section}
\numberwithin{theorem}{section}
\numberwithin{example}{section}
\newcommand{\eq}[1]{\hyperref[eq:#1]{(\ref*{eq:#1})}}
\renewcommand{\sec}[1]{\hyperref[sec:#1]{Section~\ref*{sec:#1}}}
\newcommand{\thm}[1]{\hyperref[thm:#1]{Theorem~\ref*{thm:#1}}}
\newcommand{\lem}[1]{\hyperref[lem:#1]{Lemma~\ref*{lem:#1}}}
\newcommand{\cor}[1]{\hyperref[cor:#1]{Corollary~\ref*{cor:#1}}}
\newcommand{\itm}[1]{\hyperref[itm:#1]{\ref*{itm:#1}}}
\newcommand{\fig}[1]{\hyperref[fig:#1]{Figure~\ref*{fig:#1}}}
\newcommand{\ex}[1]{\hyperref[ex:#1]{Example~\ref*{ex:#1}}}

\newcommand{\nm}[1]{\lVert #1\rVert}
\newcommand{\sem}[1]{[\![ #1 ]\!]}
\newcommand{\bra}[1]{\langle #1 \vert}
\newcommand{\ket}[1]{\vert #1 \rangle}
\newcommand{\tr}[0]{\mathrm{tr}}

\newcommand{\cskip}[0]{{\mathbf{skip}}}

\newcommand{\qif}[1]{\mathbf{case}~#1~\mathbf{end}}
\newcommand{\qwhile}[2]{\mathbf{while}~#1~\mathbf{do}~#2~\mathbf{done}}

\newcommand{\tdist}[2]{\mathrm{T}(#1,~#2)}

\newcommand{\D}[0]{\mathcal{D}}
\newcommand{\E}[0]{\mathcal{E}}
\renewcommand{\H}[0]{\mathcal{H}}
\newcommand{\F}[0]{\mathcal{F}}

\newcommand{\StepsTo}[3]{\langle #1,~\rho \rangle \rightarrow \langle #2,~#3 \rangle}
\newcommand{\StepsToP}[5]{\langle #2, #3 \rangle \xrightarrow[]{#1} \langle #4, #5 \rangle}

\newcommand{\app}[1]{\hyperref[app:#1]{Appendix~\ref*{app:#1}}}
\newcommand{\khh}[1]{}
\newcommand{\szhu}[1]{}
\newcommand{\shh}[1]{}
\newcommand{\xw}[1]{}
\newcommand{\mwh}[1]{}
\def \quwhile {quantum \textbf{while}}
\def \A {\mathcal{A}}
\def \ideal {\mathtt{ideal}}

\newcommand{\rjudgment}[4]{\ensuremath{(#1, #2) \vdash #3 \leq #4}} % robustness judgment notation
\begin{document}

%% Title information
% \iftechrep
\title{Quantitative Robustness Analysis of Quantum Programs (Extended Version)}         %% [Short Title] is optional;
% \else
% \title{Quantitative Robustness Analysis of Quantum Programs}         %% [Short Title] is optional;
% \fi
                                        %% when present, will be used in
                                        %% header instead of Full Title.
%\titlenote{with title note}             %% \titlenote is optional;
                                        %% can be repeated if necessary;
                                        %% contents suppressed with 'anonymous'
%\subtitle{Subtitle}                     %% \subtitle is optional
%\subtitlenote{with subtitle note}       %% \subtitlenote is optional;
                                        %% can be repeated if necessary;
                                        %% contents suppressed with 'anonymous'

%% Author information
%% Contents and number of authors suppressed with 'anonymous'.
%% Each author should be introduced by \author, followed by
%% \authornote (optional), \orcid (optional), \affiliation, and
%% \email.
%% An author may have multiple affiliations and/or emails; repeat the
%% appropriate command.
%% Many elements are not rendered, but should be provided for metadata
%% extraction tools.

%% Author with single affiliation.
\author{Shih-Han Hung}
% \affiliation{
%   \institution{University of Maryland, College Park}
%   \country{USA}
% }

\author{Kesha Hietala}
% \affiliation{
%   \institution{University of Maryland, College Park}
%   \country{USA}
% }

\author{Shaopeng Zhu}
\affiliation{
  \institution{University of Maryland, College Park}
  \country{USA}
}

\author{Mingsheng Ying}
\affiliation{
  \institution{University of Technology Sydney}
  \country{Australia}
}
\affiliation{
  \institution{State Key Laboratory of Computer Science,
    Institute of Software, Chinese Academy of Sciences}
  \country{China}
}
\affiliation{
 \institution{Tsinghua University}
 \country{China}
}

\author{Michael Hicks}
% \affiliation{
%   \institution{University of Maryland, College Park}
%   \country{USA}
% }

\author{Xiaodi Wu}
\affiliation{
  \institution{University of Maryland, College Park}
  \country{USA}
}
\begin{abstract}
Quantum computation is a topic of significant recent interest, with
practical advances coming from both research and
industry. A major challenge in 
quantum programming is dealing with errors (quantum noise) during
execution. Because quantum resources (e.g., qubits) are scarce,
classical error correction techniques applied at the level of the
architecture are currently cost-prohibitive. But while this reality means
that quantum programs are almost certain to have errors, there as yet
exists no principled means to reason about erroneous behavior.
This paper attempts to fill this gap by developing a semantics for
erroneous quantum while programs, as well as a logic for reasoning
about them. This logic permits proving a property we have
identified, called $\epsilon$-robustness, which characterizes
possible ``distance'' between an ideal program and an erroneous one. We
have proved the logic sound, and showed its utility on several case
studies, notably: (1) analyzing the robustness of noisy versions of the
quantum Bernoulli factory (QBF) and quantum walk (QW); (2) demonstrating
the (in)effectiveness of different error correction schemes on
single-qubit errors; and (3) analyzing the robustness of a
fault-tolerant version of QBF.
\end{abstract}

%% 2012 ACM Computing Classification System (CSS) concepts
%% Generate at 'http://dl.acm.org/ccs/ccs.cfm'.
%\begin{CCSXML}
%<ccs2012>
%<concept>
%<concept_id>10011007.10011006.10011008</concept_id>
%<concept_desc>Software and its engineering~General programming languages</concept_desc>
%<concept_significance>500</concept_significance>
%</concept>
%<concept>
%<concept_id>10003456.10003457.10003521.10003525</concept_id>
%<concept_desc>Social and professional topics~History of programming languages</concept_desc>
%<concept_significance>300</concept_significance>
%</concept>
%</ccs2012>
%\end{CCSXML}

%\ccsdesc[500]{Software and its engineering~General programming languages}
%\ccsdesc[300]{Social and professional topics~History of programming languages}
%% End of generated code

\begin{CCSXML}
<ccs2012>
<concept>
<concept_id>10003752.10010124.10010131.10010133</concept_id>
<concept_desc>Theory of computation~Denotational semantics</concept_desc>
<concept_significance>500</concept_significance>
</concept>
<concept>
<concept_id>10003752.10003753.10003758.10010626</concept_id>
<concept_desc>Theory of computation~Quantum information theory</concept_desc>
<concept_significance>500</concept_significance>
</concept>
</ccs2012>
\end{CCSXML}

\ccsdesc[500]{Theory of computation~Denotational semantics}
\ccsdesc[500]{Theory of computation~Quantum information theory}

%% Keywords
%% comma separated list
\keywords{quantum programming, quantum noise, approximate computing}  %% \keywords are mandatory in final camera-ready submission

%% \maketitle
%% Note: \maketitle command must come after title commands, author
%% commands, abstract environment, Computing Classification System
%% environment and commands, and keywords command.
\maketitle

\renewcommand{\shortauthors}{S. Hung, K. Hietala, S. Zhu, M. Ying, M. Hicks, and X. Wu}

\section{Introduction}

Quantum programming has been actively investigated for the past two
decades. Early work on semantics and language design~\citep{Om03,
  SZ00, Sabry-Haskel, Se04, AG05} has been followed up,
in the last few years, by the development of a number of mature languages,
including Quipper \citep{Green2013}, Scaffold \citep{Sca12}, 
LIQUi$\vert\rangle$ \citep{WS2014}, Q\#~\citep{Svore:2018}, and QWIRE \citep{PRZ2017}.
Various program logics have also
been extended for verification of quantum programs \citep{BJ04, CHADHA200619,
  Baltag2011, Feng:2007, Kaku09, Yin11, YYW17}. 
For detailed
surveys, see \citet{Selinger04}, \citet{Gay:2006}, and~\citet{Yin16}.

A major practical challenge in implementing quantum programs is
dealing with errors (aka quantum noise) during execution. 
Most existing work on algorithms and programming languages assumes this problem will be solved by the hardware, as in classical computers, or with fault-tolerant protocols that are designed independently of any particular application \citep{CFM2017}.
As such, the semantics of programs is defined in a manner that ignores the possibility of errors~\citep{Green2013, PRZ2017, WS2014}. 

Unfortunately, providing such a general-purpose, fault-tolerant
quantum computing abstraction appears to be impractical for near-term
quantum devices, for which precisely controllable qubits are expensive, error-prone, and scarce. 
Existing error correction techniques consume a substantial number of
qubits, severely limiting the range of possible computations. For
example, one logical qubit may require $10^3 - 10^4$ physical qubits~\citep{FMMC12}.
Furthermore, fault-tolerant operations on these logical qubits require many more physical operations than their non-fault-tolerant counterparts.

As such, research on
\emph{practical} quantum computation must focus on Noisy
Intermediate-Scale Quantum (NISQ) computers (as phrased
by~\citet{Preskill18}), which will lack general-purpose fault
tolerance. While some particular algorithms have been developed to reflect this
reality~\citep{NC-VQE, IBM-QE}, there is as yet no principled
method to reason about the error-affected performance of quantum
applications. 
Such methods are needed to help guide the design of practical applications for near-term devices.

\vspace{1mm} \noindent \textbf{Contributions.} This paper extends the
quantum \textbf{while}-language~\citep{Yin11} with a semantics that accounts for the
possibility of error, and defines an accompanying logic for
reasoning about erroneous executions. Our work constitutes an
alternative to the common, but impractical, one-size-fits-all approach to fault
tolerance and instead elevates the question of errors to the level of
the programming language. Our approach is inspired by the work of
\citet{CMR2013}, which reasons about classical programs running on unreliable
hardware. We make four main contributions.

First, we present the syntax and semantics (both operational and
denotational) of the \quwhile-language extended to
include noisy operations. In particular, we modify unitary
application to allow the noisy operation $\Phi$ (a superoperator) to occur
with probability $p$. 
This approach permits modeling any local noise occurring during the execution of a quantum program, which is the standard noise model 
considered in the study of quantum error correction and fault-tolerant quantum computation~\citep{Got10}. 
This error model is also used by experimental physicists for building and benchmarking quantum devices in both academia and industry.

Second, we define a notion of \emph{quantum robustness}. In particular, we
say that a noisy program $\widetilde{P}$ is \emph{$\epsilon$-robust
  under $(Q,\lambda)$} if it computes a quantum state at most $\epsilon$
distance away from that of that of its ideal equivalent $P$ when
starting both from states satisfying quantum predicate $Q$~\citep{DP2006} to degree
$\lambda$. Our definition makes use of the so-called \emph{diamond
norm}~\citep{diamond-norm} to account for the potential
enlarging effect of entanglement on the distance.
We generalize the diamond norm to what we call the $(Q,
\lambda)$-diamond norm, which allows us to consider 
only input states that satisfy a quantum predicate $Q$ to degree
$\lambda$. Doing so obtains more accurate bounds when
considering specific quantum devices and/or knowledge of
states owing to the use of classical control operators.
We show that the $(Q, \lambda)$-diamond norm can be
computed by a semidefinite program (SDP) by extending the 
algorithm of~\citet{W09}.

Third, we define a logic for reasoning about
quantum robustness, with the following judgment\footnote{%
Our notation draws an analogy with the typing judgment 
$\Gamma \vdash P : t$. In particular, $Q$ and $\lambda$ 
are ``assumptions'' about inputs just as $\Gamma$ represents assumptions 
about input (their types); $\widetilde{P}$ is the program we are reasoning about; 
and $\epsilon$ is the proved robustness of this program, just as $t$ is the 
proved type of the program. 
Our rules are compositional like those of typing.
%\mwh{I'm not sure we really need this.}%
} 
\begin{equation} \label{eq:judgment}
\rjudgment{Q}{\lambda}{\widetilde{P}}{\epsilon}.
\end{equation}
This judgment states that for any input state that satisfies a quantum
predicate $Q$ to degree $\lambda$, the distance between
$\sem{\widetilde{P}}$ and $\sem{P}$ is then bounded by $\epsilon$. 

We prove our logic is sound. A particular challenge is the rule for
loops, owing to the \emph{termination problem} first studied by~\citet{LY2017}.
In particular, if the loop body generates some error and the loop does
not terminate, it is hard to prove any non-trivial bound on the final
accumulated error. To avoid this difficulty in the setting of
approximate computing, \citet{CMR2013} simply
assume that the loop will terminate within a bounded number of
iterations or a trivial upper bound will be applied.
To capture more complicated cases, we introduce a concept called the
\emph{$(a,n)$-boundedness} of the loop. Intuitively, a loop is
$(a,n)$-bounded if, for every input state, after $n$ iterations it is
guaranteed that with probability at least $1-a$ it has exited the loop. 
The probability is due to the quantum measurement in the loop guard.
 It is easy to see that $(a,n)$-boundedness implies the termination
of the loop. Pleasantly, $(a,n)$ bounding the
ideal loop is sufficient to reason about a noisy loop with any error model
$\Phi$ in the loop body.   

Our fourth and final contribution is to develop several case studies
that demonstrate the utility of reasoning about errors at the program
level. 
\begin{itemize}[leftmargin=*]
\item 
We start with important examples in the \quwhile-language, the
\emph{quantum Bernoulli factory} (QBF) and the \emph{quantum walk
  (QW)}, and directly analyze the robustness of noisy versions
$\widetilde{QBF}$ and $\widetilde{QW}$ using our logic. In the
process, we prove the $(a,n)$-boundedness of the loops in
$\widetilde{QBF}$ and $\widetilde{QW}$ using both analytical and  numerical methods. 

\item
We also demonstrate the use of our semantics to show the efficiency of
different error correction schemes. Consider the error  correction of
a single qubit in the environment where a single bit flip error
happens with probability $0< p<1/2$. We consider three schemes: (1)
$P_1$ without any error correction; (2) $P_2$ with error correction
for bit flips; (3) $P_3$ with error correction for phase flips. We
prove that their corresponding robustnesses
$\epsilon_1,\epsilon_2,\epsilon_3$ (under trivial precondition $Q=I, \lambda=0$) satisfy $\epsilon_2< \epsilon_1<\epsilon_3$.
In other words,  we conclude that an error correction scheme that is appropriate for
the error model can reduce noise in a program, while an inappropriate error correction scheme may do the opposite.

\item
Combining these ideas, we further analyze the robustness of a fault-tolerant version of QBF and demonstrate that the use of appropriate fault-tolerant gadgets makes QBF more robust. 
\end{itemize}

\vspace{1mm} \noindent \textbf{Organization.} We introduce
preliminaries about quantum information and the \quwhile-language in
\sec{preliminaries} and \sec{quantum-while-language} respectively. We
present the noisy \quwhile-language (syntax, semantics) in
\sec{noisy-progs} and define quantum robustness and a logic for
bounding it in \sec{robustness}.  We conclude the paper
with case studies in \sec{case-studies}. We discuss related work, next.

\section{Related Work}

\subsection{Reasoning about Errors in Classical Programs}

There has been significant work on reasoning about errors in classical software, which we describe here.
We believe that reasoning about errors is even more important in a quantum setting, where we can expect that errors will be prevalent and that the nature of errors will change rapidly as hardware progresses.

\subsubsection{Faulty Hardware}
One example of an error that a classical computer may experience is a \emph{transient hardware fault}.
Previous work has demonstrated that programming language techniques can help give safety guarantees even in the presence of such errors.
For example, \citet{Walker2006} present a type-theoretic framework for analyzing fault-tolerant lambda calculus in the presence of transient faults. 
Their type system guarantees that well-typed programs can tolerate any single data fault.
\citet{Perry2007} extend those ideas to typed assembly language.
Similar to this work, we present a semantics that tracks errors during computation.

\subsubsection{Program Continuity}
Our work is also similar to existing work on verifying program continuity and robustness \citep{Chaudhuri2010, Chaudhuri2011}.
In order to prove that a program is continuous or robust, it is necessary to show that the output of that program will not change significantly given small changes to the program's inputs.
We are interested in bounding the distance between the output of a noisy program and its corresponding ideal program given the same input. 

\subsubsection{Approximate and Probabilistic Computing} 
In recent years, there have been many advances in programming language support for approximate \citep{Sampson2011, Park2015, Boston2015, Baek2010, Carbin2012} and probabilistic \citep{Bornholt2014, Sampson2014} computing.
Both of these styles of computing rely on uncertainty during computation, either due to hardware errors or randomness, but still require that programs satisfy some correctness properties. 
Programming language tools in this area have aimed to provide these correctness guarantees.

\citet{CMR2013} present a tool for verifying \emph{reliability conditions},
which indicate the probability that a computation produces the
correct result. 
Given a hardware specification, which lists the probability with which
an operation executes correctly, their analysis computes a
conservative probability that a program value is computed correctly.
We extend this notion to a quantum setting by computing the \emph{robustness} of a quantum program, which relates to the probability that the state of the system after executing the program is correct.
Also, our hardware specification records not only the probability of
error, but also the nature of the errors that may occur (operator
$\Phi$). 
Doing so is more computationally expensive, but necessary to be able to reason about the
effects of error correction schemes in quantum
programs. \citet{CMR2013} does not consider error correction; instead,
errors are considered permanent and only probabilities of
errors are used for characterizations. 

\subsection{Characterizing Error in Quantum Programs}

To our knowledge, we are presenting the first semantics and logic for quantum computation that considers errors.
Existing work on characterizing error in quantum programs has focused primarily on dynamic approaches such as simulation \citep{Gutierrez2013} or physical experimentation \citep{Chuang1997, Knill2008, Emerson2005, Magesan2011}.
Resource estimation tools like QuRE can statically produce error estimates for algorithms given hardware specifications~\citep{Suchara2013}.
However, these tools are targeted at quantum \emph{circuits} (i.e. quantum programs without conditionals or loops)  and typically assume that a single error model and single type of error correction will be used throughout the circuit.

\section{Quantum Information: Preliminaries and Notations} \label{sec:preliminaries}

\begin{table}
  \caption{A brief summary of notation used in this paper}
  \label{tab:notation}
  \begin{tabular}{@{} lll @{}}\\
    \textbf{Hilbert Spaces:} & $\H$, $\A$ & \\
    \textbf{States:} & (pure states) & $\ket{\psi}, \ket{\phi}$
    (metavariables); $\ket{0}, \ket{1}, \ket{+},\ket{-}$ (notable states) \\
  & (density operators) & $\rho, \sigma$ (metavariables);
  $\ket{\psi}\bra{\psi}$ (as outer product)\\
    \textbf{Operations:} & (unitaries) & $U, V$ (metavariables); $H, X,
    Z$ (notable operations) \\
  & (superoperators) & $\E, \F$ (general); $\Phi$ (used to represent noise)\\
    \textbf{Measurements:} & $M$ & $\{M_m\}_m$ (general);
    $\{M_0=\ket{0}\bra{0}, M_1=\ket{1}\bra{1}\}$ (example)\\
  \end{tabular}
\end{table}

This section presents background and notation on quantum information
and quantum computation. For an extended background, we recommend notes
by~\citet{Wat06} and the textbook by~\citet{MI2002}. A summary of 
notation we use appears in Table~\ref{tab:notation}.   

\subsection{Preliminaries}

For any finite integer $n$, an $n$-dimensional Hilbert space $\H$
is essentially the space $\mathbb{C}^n$ of complex vectors.
We use Dirac's notation, $\ket{\psi}$, to denote a complex vector in $\mathbb{C}^n$. The inner product of two vectors $\ket{\psi}$ and $\ket{\phi}$ is denoted by $\langle\psi|\phi\rangle$,
which is the product of the Hermitian conjugate of $\ket{\psi}$, denoted by $\bra{\psi}$, and vector $\ket{\phi}$.
The norm of a vector $\ket{\psi}$ is denoted by $\nm{\ket{\psi}}=\sqrt{\langle\psi|\psi\rangle}$.

We define (linear) \emph{operators} as linear mappings between Hilbert spaces.
Operators between $n$-dimensional Hilbert spaces are represented by $n\times n$ matrices.
For example, the identity operator $I_\H$ can be identified by the identity matrix on $\H$.
The Hermitian conjugate of operator $A$ is denoted by $A^\dag$. Operator $A$ is \emph{Hermitian} if $A=A^\dag$.
The trace of an operator $A$ %on an $n$-dimensional Hilbert space
is the sum of the entries on the main diagonal, i.e., $\tr(A)=\sum_i A_{ii}$. %where $\{\ket{i}\}$ is any orthonormal basis.
We write $\bra{\psi}A\ket{\psi}$ to mean the inner product between
$\ket{\psi}$ and $A\ket{\psi}$.
A Hermitian operator $A$ is \emph{positive semidefinite} (resp.,
\emph{positive definite}) if for all vectors $\ket{\psi}\in\H$,
$\bra{\psi}A\ket{\psi}\geq 0$ (resp., $>0$).
This gives rise to the \emph{L\"owner order} $\sqsubseteq$ among operators:
\begin{equation}
 %\text{(L\"owner order)  }
 A\sqsubseteq B \text{ if } B-A \text{ is positive semidefinite, } \quad A\sqsubset B \text{ if } B-A \text{ is positive definite. }
\end{equation}

\subsection{Quantum States}

The state space of a quantum system is a Hilbert space. 
The state space of a \emph{qubit}, or quantum bit, is a 2-dimensional Hilbert space.
One important orthonormal basis of a qubit system is the \emph{computational} basis with $\ket{0}=(1,0)^\dag$ and $\ket{1}=(0,1)^\dag$, which encode the classical bits 0 and 1 respectively.  
Another important basis, called the $\pm$ basis, consists of $\ket{+}=\frac{1}{\sqrt{2}}(\ket{0}+\ket{1})$ and $\ket{-}=\frac{1}{\sqrt{2}}(\ket{0}-\ket{1})$.
The state space of multiple qubits is the \emph{tensor product} of single qubit state spaces.
For example, classical 00 can be encoded by $\ket{0}\otimes\ket{0}$
(written $\ket{0}\ket{0}$ or even $\ket{00}$ for short) in the Hilbert space $\mathbb{C}^2\otimes\mathbb{C}^2$. 
The Hilbert space for an $m$-qubit system is $(\mathbb{C}^2)^{\otimes m} \cong \mathbb{C}^{2^m}$.

A \emph{pure} quantum state is represented by a unit vector, i.e., a vector $\ket{\psi}$ with $\nm{\ket{\psi}}=1$.
A \emph{mixed} state can be represented by a classical distribution over an ensemble of pure states $\{(p_i,\ket{\psi_i})\}_i$,
i.e., the system is in state $\ket{\psi_i}$ with probability $p_i$.
One can also use \emph{density operators} to represent both pure and mixed quantum states.
A density operator $\rho$ for a mixed state representing the ensemble $\{(p_i,\ket{\psi_i})\}_i$ is a positive semidefinite operator $\rho=\sum_i p_i\ket{\psi_i}\bra{\psi_i}$, where $\ket{\psi_i}\bra{\psi_i}$ is the outer-product of $\ket{\psi_i}$; in particular, a pure state $\ket{\psi}$ can be identified with the density operator $\rho=\ket{\psi}\bra{\psi}$.
Note that $\tr(\rho)=1$ holds for all density operators. A positive semidefinite operator $\rho$ on $\H$ is said to be a \emph{partial} density operator if $\tr(\rho)\leq 1$.
The set of partial density operators is denoted by $\D(\H)$.
%We will denote by $\D(\H)$ the set of all partial density matrices on $\H$.

\subsection{Quantum Operations} 

Operations on quantum systems can be characterized by unitary operators. An operator $U$ is \emph{unitary} if its Hermitian conjugate is its own inverse, i.e., $U^\dag U=UU^\dag=I$. For a pure state $\ket{\psi}$, a unitary operator describes an \emph{evolution} from $\ket{\psi}$ to $U\ket{\psi}$. For a density operator $\rho$, the corresponding evolution is $\rho \mapsto U\rho U^\dag$. Common single-qubit unitary operators include
\begin{align}
  H=\frac{1}{\sqrt{2}}
  \left[ {\begin{array}{cc}
   1 & 1 \\
   1 & -1 \\
  \end{array} } \right],
  \quad
  X=
  \left[ {\begin{array}{cc}
   0 & 1 \\
   1 & 0 \\
  \end{array} } \right],
  \quad
  Z=
  \left[ {\begin{array}{cc}
   1 & 0 \\
   0 & -1 \\
  \end{array} } \right].
\end{align}
The \emph{Hadamard} operator $H$ transforms between the computational and the $\pm$ basis. For example, $H\ket{0}=\ket{+}$ and $H\ket{1}=\ket{-}$.
The \emph{Pauli $X$} operator is a bit flip, i.e., $X\ket{0}=\ket{1}$ and  $X\ket{1}=\ket{0}$. The \emph{Pauli $Z$} operator is a phase flip, i.e., $Z\ket{0}=\ket{0}$ and $Z\ket{1}=- \ket{1}$.

More generally, the evolution of a quantum system can be characterized by an \emph{admissible superoperator} $\E$, which is a \emph{completely-positive} and \emph{trace-non-increasing} linear map from $\D(\H)$ to $\D(\H')$ for Hilbert spaces $\H, \H'$.
A superoperator is positive if it maps from $\D(\H)$ to $\D(\H')$ for Hilbert spaces $\H, \H'$.
A superoperator $\E$ is $k$-positive if for any $k$-dimensional Hilbert space $\A$, the superoperator $\E\otimes I_\A$ is a positive map on $\D(\H\otimes\A)$. 
A superoperator is said to be completely positive if it is $k$-positive for any positive integer $k$.
A superoperator $\E$ is trace-non-increasing if for any initial state $\rho\in\D(\H)$, the final state $\E(\rho)\in \D(\H')$ after applying $\E$ satisfies $\tr(\E(\rho))\leq\tr(\rho)$.

For every superoperator $\E : \D(\H)\to\D(\H')$, there exists a set of Kraus operators $\{E_k\}_k$ such that $\E(\rho)=\sum_k E_k\rho E_k^\dag$ for any input $\rho\in\D(\H)$.
Note that the set of Kraus operators is finite if the Hilbert space is finite-dimensional.
The \emph{Kraus form} of $\E$ is written as $\E=\sum_k E_k\circ E_k^\dag$.
A unitary evolution can be represented by the superoperator $\E=U\circ U^\dag$. An identity operation refers to the superoperator $\mathcal{I}_{\H} = I_{\H} \circ I_{\H}$.
The Schr\"odinger-Heisenberg \emph{dual} of a superoperator $\E=\sum_k E_k\circ E_k^\dag$, denoted by $\E^*$, is defined as follows: for every state $\rho\in\D(\H)$ and any operator $A$, $\tr(A\E(\rho))=\tr(\E^*(A)\rho)$. The Kraus form of $\E^*$ is $\sum_k E_k^\dag\circ E_k$.

\subsection{Quantum Measurements} 

The way to extract information about a quantum system is called
a quantum \emph{measurement}. 
A quantum measurement on a system over Hilbert space $\H$ can be
described by a set of linear operators 
$\{M_m\}_m$ with $\sum_m M_m^\dag M_m=I_\H$. 
If we perform a measurement $\{M_m\}$ on a state $\rho$, the outcome $m$ is observed with probability $p_m=\tr(M_m\rho M_m^\dag)$ for each $m$.
A major difference between classical and quantum computation is that a
quantum measurement changes the state. In particular, after a
measurement yielding outcome $m$, the state collapses to $M_m\rho M_m^\dag/p_m$.
For example, a measurement in the computational basis is described by $M=\{M_0=\ket{0}\bra{0}, M_1=\ket{1}\bra{1}\}$.
If we perform the computational basis measurement $M$ on state $\rho=\ket{+}\bra{+}$, then with probability $\frac{1}{2}$ the outcome is $0$ and $\rho$ becomes $\ket{0}\bra{0}$.
With probability $\frac{1}{2}$ the outcome is $1$ and $\rho$ becomes $\ket{1}\bra{1}$.

\section{Quantum Programs}
\label{sec:quantum-while-language}
Our work builds on top of the \quwhile-language developed by \citet{Yin11,Yin16}.
Here we review the syntax and semantics of this language.

\subsection{Syntax}
\label{sec:syntax}

Define $\mathit{Var}$ as the set of quantum variables. 
We use the symbol $q$ as a metavariable ranging over quantum variables and define a \emph{quantum register} $\overline{q}$ to be a finite set of distinct variables.
For each $q\in\mathit{Var}$, its state space is denoted by $\mathcal{H}_q$.
The quantum register $\overline{q}$ is associated with the Hilbert space $\H_{\overline{q}}=\bigotimes_{q\in\overline{q}}\H_q$.
If $type(q) =$ \textbf{Bool} then $\H_q$ is the two-dimensional Hilbert space with basis $\{\ket{0}, \ket{1}\}$.
If $type(q) =$ \textbf{Int} then $\H_q$ is the Hilbert space with basis $\{\ket{n} : n \in \mathbb{Z}\}$.
The syntax of a \emph{\quwhile ~ program} $P$ is defined as follows.
\begin{align}
 P \enskip ::= & \quad \cskip 
                \enskip | \enskip q:=\ket{0} 
                \enskip | \enskip \overline{q}:=U[\overline{q}] 
                \enskip | \enskip P_1;P_2 \enskip | \nonumber  \\
              & \quad \qif{M[\overline{q}]=\overline{m\to P_m}}
                \enskip | \enskip \qwhile{M[\overline{q}]=1}{P_1}
\end{align}
The language constructs above are similar to their classical counterparts.
(1) $\cskip$ does nothing.
(2) $q:=\ket{0}$ sets quantum variable $q$ to the basis state $\ket{0}$.
(3) $\overline{q}:=U[\overline{q}]$ applies the unitary $U$ to the qubits in $\overline{q}$.
(4) Sequencing has the same behavior as its classical counterpart.
(5) $\qif{M[\overline{q}]=\overline{m \to P_m}}$ performs the measurement $M = \{M_m\}$ on the qubits in $\overline{q}$, and executes program $P_m$ if the outcome of the measurement is $m$.
The bar over $\overline{m \to P_m}$ indicates that there may be one or more repetitions of this expression.
\footnote{Our syntax for conditional/case statements differs from that presented by \citet{Yin16} to make it more clear that there are multiple programs $P_m$.}
(6) $\qwhile{M[\overline{q}]=1}{P_1}$ performs the measurement $M = \{M_0, M_1\}$ on the qubits in $\overline{q}$, and executes $P_1$ if measurement produces the outcome corresponding to $M_1$ or terminates if measurement produces the outcome corresponding to $M_0$.

We highlight two differences between quantum and classical while languages:
(1) Qubits may only be initialized to the basis state $\ket{0}$. There is no quantum analogue for initialization to any expression (i.e. $x:=e$) because of the no-cloning theorem of quantum states.
Any state $\ket{\psi} \in \H_q$, however, can be constructed by applying some unitary $U$ to $\ket{0}$.
\footnote{In our examples, we may write $q:=\ket{\psi}$ for some fixed basis state $\ket{\psi}$. 
What we mean in this case is $q:=\ket{0}; q:=U[q]$ where $U$ is the unitary operation that transforms $\ket{0}$ into $\ket{\psi}$.}
(2) Evaluating the guard of a case statement or loop, which performs a measurement, potentially disturbs the state of the system.

We now present an example program written in the \quwhile-language syntax.
The \emph{quantum walk} \citep{AAKV2001} is a widely considered example in quantum programming, quantum algorithms, and quantum simulation literature \citep{GAN2014,YYW17}.
Here we consider a quantum walk on a circle with $n$ points.
We let the initial position of the walker be $0$, and say that the program halts if and only if the walker arrives at position $1$.

\begin{example}[Quantum Walk]
Define the coin (or "direction") space $\mathcal{H}_c$ to be the $2$-dimensional Hilbert space with orthonormal basis states $\ket{L}$ and $\ket{R}$, for \emph{Left} and \emph{Right} respectively.
Define the position space $\mathcal{H}_p$ to be the $n$-dimensional Hilbert space with orthonormal basis states $\ket{0}, \ket{1}, ..., \ket{n-1}$, where vector $\ket{i}$ represents position $i$ for $0 \leq i < n$.
Now the state space of the walk is $\mathcal{H} = \mathcal{H}_c \otimes \mathcal{H}_p$ and the initial state is $\ket{L}\ket{0}$.
In each step of the walk:
\begin{enumerate}
\item Measure the position of the system to determine whether the walker has reached position $1$.
If the walker has reached position $1$, the walk terminates.
Otherwise, it continues.
We use the measurement $M = \{\ket{1}\bra{1}, \sum_{i \neq 1}\ket{i}\bra{i}\}$.
\item Apply the ``coin-tossing'' operator $H$ to the coin space $\mathcal{H}_c$.
\item Perform the shift operator $S$ defined by
$S\ket{L,i} = \ket{L,i - 1 (\text{mod}~n)},~~~S\ket{R,i} = \ket{R,i + 1 (\text{mod}~n)}$
for $i = 0,1,...,n-1$ to the space $\mathcal{H}$.
The $S$ operator can be written as
\[ S = \sum_{i=0}^{n-1}\ket{L}\bra{L} \otimes \ket{i - 1 (\text{mod}~n)}\bra{i} + \sum_{i=0}^{n-1}\ket{R}\bra{R} \otimes \ket{i + 1 (\text{mod}~n)}\bra{i}.\]
\end{enumerate}
In this algorithm, the walker takes one step left or one step right corresponding to the coin flip result $\ket{L}$ or $\ket{R}$.
However, unlike the classical case, the result of the coin flip may be a superposition of $\ket{L}$ and $\ket{R}$, allowing the walker to take a step to the left and right \emph{simultaneously}.
This quantum walk can be described by the following program
\begin{equation} \label{eqn:QW_N}
 QW_n\equiv p:=\ket{0};c:=\ket{L};\qwhile{M[p] = 1}{c:=H[c]; c,p:=S[c,p]}.
\end{equation}
\end{example}

\begin{figure}

  \begin{subfigure}[b]{\textwidth}
    \begin{align*}
      \text{(Skip)}& & &\infer[]{\StepsTo{\cskip}{E}{\rho}}{} \\
      \text{(Initialization)}& & &\infer[]{\StepsTo{q:=\ket{0}}{E}{\rho_0^q}}{}\\
      & & &\text{where}~\rho_0^q = \begin{cases}
                   \E_{q \rightarrow 0}^{\mathrm{bool}}(\rho)  & \text{if $type(q)$ = \textbf{Bool}}\\
                   \E_{q \rightarrow 0}^{\mathrm{int}}(\rho) & \text{if $type(q)$ = \textbf{Int}}
              \end{cases}  \\
      \text{(Unitary)}& & &\infer[]{\StepsTo{\overline{q}:=U[\overline{q}]}{E}{U\rho U^\dagger}}{}\\
      \text{(Sequence E)}& & &\infer[]{\StepsTo{E;P_2}{P_2}{\rho}}{}\\
      \text{(Sequence)}& & &\infer[]{\StepsTo{P_1;P_2}{P_1';P_2}{\rho'}}{\StepsTo{P_1}{P_1'}{\rho'}}\\
      \text{(Case $m$)}& & &\infer[]{\StepsTo{\qif{M[\overline{q}]=\overline{m\to P_m}}}{P_m}{M_m \rho M_m^\dagger}}{}\\
      & & &\text{for each outcome $m$ of measurement $M = \{ M_m \}$}\\
      \text{(While 0)}& & &\infer[]{\StepsTo{\qwhile{M[\overline{q}]=1}{P_1}}{E}{M_0 \rho M_0^\dagger}}{} \\
      \text{(While 1)}& & &\infer[]{\StepsTo{\qwhile{M[\overline{q}]=1}{P_1}}{P_1; \qwhile{M[\overline{q}]=1}{P_1}}{M_1 \rho M_1^\dagger}}{}
    \end{align*}
    \caption{}
    \label{fig:opsem}
  \end{subfigure}

  \begin{subfigure}[b]{\textwidth}
    \begin{align*}
      \begin{array}{lll}
        \sem{\cskip}\rho & = & \rho \\
        \sem{q:=\ket{0}}\rho & = & \begin{cases}
             \E_{q \rightarrow 0}^{\mathrm{bool}}(\rho)& \text{if $type(q)$ = \textbf{Bool}}\\
            \E_{q \rightarrow 0}^{\mathrm{int}}(\rho)  & \text{if $type(q)$ = \textbf{Int}}
        \end{cases} \\
        \sem{\overline{q}:= U[\overline{q}]}\rho & = & U\rho U^\dag \\
        \sem{P_1;P_2}\rho & = & \sem{P_2}(\sem{P_1}\rho) \\
        \sem{\qif{M[\overline{q}]=\overline{m\to P_m}}}\rho & = & \sum_m \sem{P_m}(M_m\rho M_m^\dag) \\
        \sem{\qwhile{M[\overline{q}]=1}{P_1}}\rho & = & \bigsqcup_{k=0}^{\infty}\sem{\mathbf{while}^{(k)}}\rho
      \end{array}
    \end{align*}
    \caption{}
    \label{fig:desem}
  \end{subfigure}

  \caption{
  \quwhile ~ programs: (a) operational semantics
  (b) denotational semantics.}
\end{figure}

\subsection{Operational Semantics}

The \emph{operational} semantics of the \quwhile-language are presented in \fig{opsem}.
$\StepsTo{P}{P'}{\rho'}$, where $\langle P,~\rho \rangle$ and $\langle P',~\rho' \rangle$ are quantum \emph{configurations}. 
In configurations, $P$ (or $P'$) could be a quantum program or the
empty program $E$, and $\rho$ and $\rho'$ are partial density
operators representing the current state. 
Intuitively, in one step, we can evaluate program $P$ on input state $\rho$ to program $P'$ (or $E$) and output state $\rho'$.
In order to present the rules in a non-probabilistic manner, the probabilities associated with each transition are encoded in the output partial density operator.\footnote{If we had instead considered a probabilistic transition system, then
the transition rule for case statements could have been written as
$\StepsToP{p_m}{\qif{M[\overline{q}]=\overline{m\to P_m}}}{\rho}{P_m}{\rho_m}$
where $p_m = \tr(M_m\rho M_m^\dagger)$ and $\rho_m = M_m\rho M_m^\dagger / p_m$.}

In the \emph{Initialization} rule, the superoperators  $\E_{q \rightarrow 0}^{\mathrm{bool}}(\rho)$ and $\E_{q \rightarrow 0}^{\mathrm{int}}(\rho)$, which initialize the variable $q$ in $\rho$ to $\ket{0}\bra{0}$, are defined by 
$\E_{q \rightarrow 0}^{\mathrm{bool}}(\rho) = \ket{0}_q\bra{0}\rho\ket{0}_q\bra{0} + \ket{0}_q\bra{1}\rho\ket{1}_q\bra{0}$ and  $\E_{q \rightarrow 0}^{\mathrm{int}}(\rho)=\sum_{n=-\infty}^{\infty}\ket{0}_q\bra{n}\rho\ket{n}_q\bra{0}$. 
Here, $\ket{\psi}_q\bra{\phi}$ denotes the outer product of states $\ket{\psi}$ and $\ket{\phi}$ associated with variable $q$; that is, $\ket{\psi}$ and $\ket{\phi}$ are in $\H_q$ and $\ket{\psi}_q\bra{\phi}$ is a matrix over $\H_q$. 
It is a convention in the quantum information literature that when operations or measurements only apply to part of the quantum system (e.g., a subset of quantum variables of the program), one should assume that an identity operation is applied to the rest of quantum variables.  
For example, applying $\ket{\psi}_q\bra{\phi}$ to $\rho$ means applying $\ket{\psi}_q\bra{\phi}\otimes I_{\H_{\bar{q}}}$ to $\rho$, where $\bar{q}$ denotes the set of all variables except $q$. 
The identity operation is usually omitted for simplicity. 

We do not explain the rules in detail, but hope their meaning is
self-evident given the description of the language in \sec{syntax}.

\begin{figure}
  \centering
  \includegraphics[trim={8.5cm 5.5cm 1.5cm 4.5cm},clip,width=7.5cm]{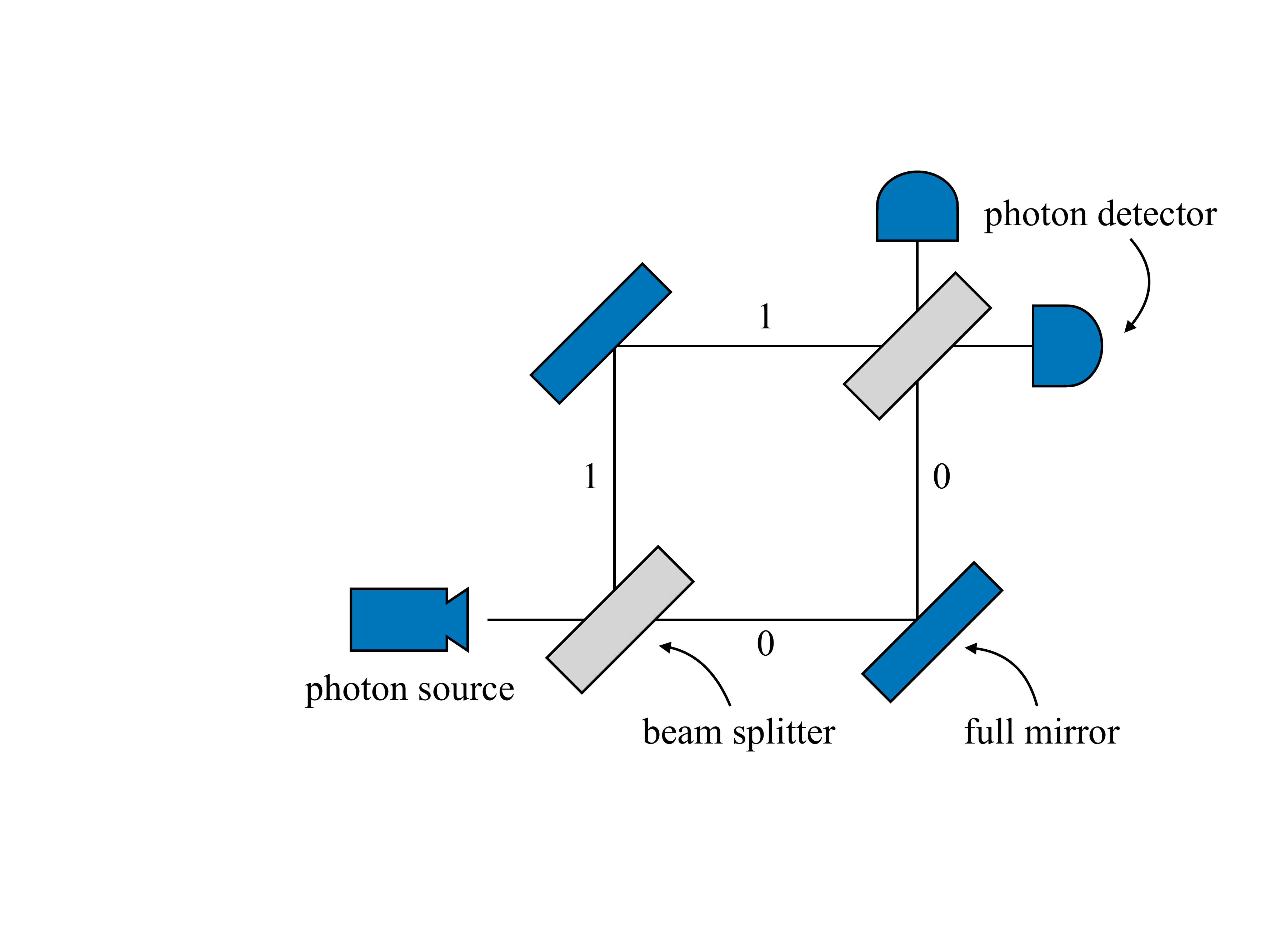}
  \caption{The beam splitter experiment.
  The `0' path corresponds to the photon having been transmitted at the first beam splitter, and the `1' path corresponds to the photon having been reflected at the first beam splitter.}
  \label{fig:beam-splitter}
\end{figure}

To illustrate the use of the operational semantics, we consider the classic \emph{beam splitter experiment}, which demonstrates the difference between quantum and classical mechanics.
In this experiment, as shown in \fig{beam-splitter}, a photon source sends photons through two beam splitters.
The final locations of the photons are determined using photon detectors.
A classical analysis would assume that, for each photon, each beam splitter flips a fair coin to either reflect the photon or allow it to pass through.
The full mirrors reflect incoming photons with probability one.
As a result, we might expect half of the photons to reach one detector and half to reach the other.
However, this is not what is observed in experiments. On the contrary, all photons will reach one detector.

\begin{example}[Beam Splitter Experiment]
  Model each beam splitter as a Hadamard gate and say that the photon source produces photons on the `0' path, which corresponds to the $\ket{0}$ state.
  Then the beam splitter experiment (BSE) corresponds to the program
  \begin{equation} \label{eqn:BSE}
     BSE \equiv q_1:=\ket{0}; q_1:=H[q_1]; q_1:=H[q_1]
  \end{equation}
  where $type(q_1)$ = \textbf{Bool}.
  Let $\rho = \ket{1}_{q_1}\bra{1}$.
  Then the evaluation of program $BSE$ on input $\rho$ proceeds as follows.
  \begin{align*}
    \langle BSE, \rho \rangle
     &= \langle q_1:=\ket{0}; q_1:=H[q_1]; q_1:=H[q_1],~~~~\ket{1}_{q_1}\bra{1} \rangle\\
     &\to \langle q_1:=H[q_1]; q_1:=H[q_1],~~~~\ket{0}_{q_1}\bra{0} \rangle \\
     &\to \langle q_1:=H[q_1],~~~~\ket{+}_{q_1}\bra{+} \rangle\\
     &\to \langle E,~~~~\ket{0}_{q_1}\bra{0} \rangle
  \end{align*}
  At the first beam splitter, the photons may either continue on their current path (corresponding to $\ket{0}$)
  or be reflected to the `1' path (corresponding to $\ket{1}$).
  In quantum mechanics, both possibilities happen simultaneously, resulting in each photon continuing on a superposition of both paths.
  The superposition of both paths corresponds to the state $\ket{+} = \frac{1}{\sqrt{2}}(\ket{0} + \ket{1})$.
  At the second beam splitter, the paths in superposition interfere with each other, resulting in the `1' path being cancelled.
\end{example}

\subsection{Denotational Semantics}
The \emph{denotational} semantics of a  \quwhile ~ program is given in
\fig{desem}. It defines $\sem{P}$ as a superoperator that acts on
$\rho \in \H_\mathit{Var}$~\cite{Yin16}. The semantics of each term is
given compositionally. We write $\mathbf{while}^{(k)}$ for the kth
syntactic approximation (i.e., unrolling) of $\mathbf{while}$ and
$\bigsqcup$ for the \emph{least upper bound} operator in the complete
partial order generated by L\"{o}wner comparison. 
For more detail on the semantics of loops, we refer the reader to \citet{Yin11,Yin16}.

We connect the denotational semantics to the operational
semantics through the following proposition.

\begin{proposition}[\cite{Yin16}]
For any program $P$
\begin{equation} \label{eq:denotational}
 \sem{P}\rho \equiv \sum\{ \vert \rho' : \langle P, \rho \rangle \rightarrow^* \langle E, \rho' \rangle \vert \},
\end{equation}
where $\rightarrow^*$ is the reflexive, transitive closure of
$\rightarrow$ and $\{ \vert \cdot \vert \}$ denotes a multi-set. 
\end{proposition}
In short, the meaning of running program $P$ on
input state $\rho$ is the sum of all possible output states, weighted
by their probabilities.

The semantics presented so far assume that no noise will occur during computation.
In \sec{noisy-progs}, we extend the semantics to include possible errors that may occur during unitary application. 

\subsection{Quantum Predicates and Hoare Logic} \label{sec:qpredicate}

A \emph{quantum predicate}  is a Hermitian operator $M$ such that $0 \sqsubseteq M \sqsubseteq I$~\citep{DP2006}. 
For a predicate $M$ and state $\rho$,  $\tr(M\rho)$ is the expectation of the truth value of predicate $M$ in state $\rho$. 
Restricting $M$ to be between $0$ and $I$ ensures that $0 \leq \tr(M\rho) \leq 1$ for any $\rho\in\D(\H)$.

The identity matrix corresponds to the \textbf{true} predicate because for any density operator $\rho$, $\tr(I\rho) = 1$.
The zero matrix corresponds to the \textbf{false} predicate because for any density operator $\rho$, $\tr(0\rho) = 0$.
$\ket{0}\bra{0}$ is the predicate that says that a state is in the subspace spanned by $\ket{0}$.
As an example, the density operator $\rho_0$ corresponding to the state $\ket{0}$ is such that $\tr(\ket{0}\bra{0}\rho_0) = 1$,
and the density operator $\rho_1$ corresponding to the state $\sqrt{1/3}\ket{0} + \sqrt{2/3}\ket{1}$ is such that $\tr(\ket{0}\bra{0}\rho_1) = \frac{1}{3}$.

\citet{Yin11,Yin16} uses quantum predicates as the basis for defining quantum preconditions and postconditions in his quantum Hoare logic.
Let $M$ and $N$ be quantum predicates and let $P$ be a \quwhile ~ program.
Then $M$ is a precondition of $N$ with respect to $P$, written $\{M\}P \{N\}$, if
\begin{equation} \label{eqn:predicate}
\forall \rho.~ \tr(M\rho) \leq \tr(N\sem{P}\rho).
\end{equation}
This inequality can be seen as the probabilistic version of the following statement:
if state $\rho$ satisfies predicate $M$, then after applying the program $P$ the resulting state will satisfy predicate $N$.
If we include an auxiliary space $\A$, then the equivalent statement is
\begin{equation} \label{eqn:predicate-aux}
\forall \rho.~ \tr((M \otimes I_\A)\rho) \leq \tr((N \otimes I_\A)(\sem{P} \otimes I_\A) (\rho)).
\end{equation}

\section{Noisy quantum programs} \label{sec:noisy-progs} In this
section, we present the syntax and semantics for the \quwhile-language
with noise, as an extension of the \quwhile-language.  
Our syntax allows one to explicitly
encode any error model that describes local noise during the execution
of a quantum program.

\subsection{Noise in Quantum Computation}
Here we briefly discuss how noise is modeled in the study of quantum error correction and fault-tolerant quantum computation~\citep{Got10}, which in turn comes from the noise model in quantum physical experiments. 
It is a convention to only consider ``local'' noise rather than correlated noise, because benign white noise is more likely than ``adversarial'' noise in actual quantum devices. 
A few types of natural local noise arise in realistic quantum systems, which generalize classical bit-flip errors, including: 
\begin{itemize}
 \item The \emph{bit flip} noise flips the state with probability $p$, and can be represented by
 \begin{equation}
     \Phi_{p, \mathrm{bit}} = (1-p) I \circ I + p X \circ X.
 \end{equation}
 \item The \emph{phase flip} noise flips the phase with probability $p$, and can be represented by
  \begin{equation}
     \Phi_{p, \mathrm{phase}} = (1-p) I \circ I + p Z \circ Z.
 \end{equation}
\end{itemize}
Other types of noise include depolarization, amplitude damping, and phase damping~\citep{MI2002}.  

This model of noise is used by experimental physicists for building
and benchmarking quantum devices in both academia and industry~\citep{RevModPhys.87.307}. Noisy information of specific quantum devices can also be publicly available (e.g., the IBM Q-experience\footnote{\url{https://www.research.ibm.com/ibm-q/technology/devices/}.}). 

\subsection{Syntax}
The syntax of a noisy quantum program $\widetilde{P}$ is defined as follows
\begin{align}
 \widetilde{P} \enskip ::= & \quad \cskip 
                           \enskip | \enskip q:=\ket{0} 
                           \enskip | \enskip \overline{q}:\cong_{p,\Phi} U[\overline{q}]
                           \enskip | \enskip \widetilde{P_1};\widetilde{P_2} \enskip | \nonumber  \\
                           & \quad \qif{M[\overline{q}]=\overline{m\to \widetilde{P_m}}} 
                           \enskip | \enskip \qwhile{M[\overline{q}]=1}{\widetilde{P_1}}
\end{align}
This syntax is identical to that of the standard quantum while language described in \sec{quantum-while-language}, except that we have annotated the unitary application construct with an error probability $p$ and an error model $\Phi$, which is the superoperator of the noisy operation. 
The statement $\overline{q}:\cong_{p,\Phi} U[\overline{q}]$ will apply the correct operation $U$ on $\overline{q}$ with probability $1-p$ and will apply the noisy (or erroneous)  operation $\Phi$ on $\overline{q}$ with probability $p$. The nature of $\Phi$ will depend on the underlying hardware, and is a parameter to our language.

For any noisy program $\widetilde{P}$, its corresponding \emph{ideal}
program can be obtained by simply replacing any noisy unitary
operations by their ideal versions (i.e., we ignore $p$ and
$\Phi$). We write $\ideal(\widetilde{P})$ for this program, or simply
$P$ when there is no ambiguity.

We remark that noisy unitary operations are already expressive enough to capture many types of noise. First, any noise can depend on the quantum state of the system by the nature of modeling it as a quantum operation $\Phi$. 
Second, noisy initialization can be modeled as initialization followed by application of a noisy identity operation,
and noisy measurement can be modeled as application of a noisy identity operation followed by measurement.
Third, errors that occur between applications of subsequent unitaries can also be modeled by noisy identity operations.

\subsection{Semantics}

The operational and denotational semantics of noisy \quwhile ~ programs are also identical to those of the standard \quwhile ~ programs,
except that the rules related to unitary application now include an error term, as shown
in \fig{noisy-opsem} and \fig{noisy-desem}.
Note that we do not require $p$ and $\Phi$ to be the same in every instance of a unitary application. 
They may depend on the type of unitary being applied, or on other features of the program such as the number of operations performed so far.
We use \emph{explicit} characterizations of the noise given by $\Phi$ to enable us to argue about the effect of error-correcting gadgets explicitly written in the program. 
If we only considered error probability (like~\citet{CMR2013}) then we could only reason about error as accumulating throughout the program, and we would not be able to show that error-correcting gadgets reduce noise by cancelling previous errors. 

%%%% Something for ourselves to remember !!! 
%\xw{different $p$ and $\Phi$ for different $U$ is simple; harder case is $p, \Phi$ could depend on the variable in the program}

\begin{figure}

  \begin{subfigure}[b]{\textwidth}
    \begin{align*}
      \text{(Unitary-Noisy)}& & &\infer[]{\StepsTo{\overline{q}:\cong_{p,\Phi}U[\overline{q}]}{E}{(1-p)U\rho U^\dagger + p\,\Phi(\rho)}}{}
    \end{align*}
    \caption{}
    \label{fig:noisy-opsem}
  \end{subfigure}

  \begin{subfigure}[b]{\textwidth}
    \begin{align*}
      \sem{\overline{q}:\cong_{p,\Phi} U[\overline{q}]}\rho = (1-p)U\rho U^\dag+p\,\Phi(\rho)
    \end{align*}
    \caption{}
    \label{fig:noisy-desem}
  \end{subfigure}

  \caption{
  (a) Transition rule for noisy unitary application.
  (b) Denotation of noisy unitary application.}
  \label{fig:noisy}
\end{figure}

\begin{example}[Beam Splitter Experiment with Errors]
  Consider the following noisy version of the beam splitter program.
  \[ \widetilde{BSE} \equiv q_1:=\ket{0}; q_1:\cong_{p, \Phi}H[q_1]; q_1:\cong_{p, \Phi}H[q_1]. \]
  Note that the error probability $p$ and error model $\Phi$ are the same in both applications of $H$.
  Let $\rho = \ket{1}_{q_1}\bra{1}$.
  Then the evaluation of program $\widetilde{P}$ on input $\rho$ proceeds as follows.
  \begin{align*}
    \langle \widetilde{BSE}, \rho \rangle
     &= \langle q_1:=\ket{0}; q_1:\cong_{p, \Phi}H[q_1]; q_1:\cong_{p, \Phi}H[q_1],~~~~\ket{1}_{q_1}\bra{1} \rangle\\
     &\to \langle q_1:\cong_{p, \Phi}H[q_1]; q_1:\cong_{p, \Phi}H[q_1],~~~~\ket{0}_{q_1}\bra{0} \rangle \\
     &\to \langle q_1:\cong_{p, \Phi}H[q_1],~~~~(1-p)\ket{+}_{q_1}\bra{+} + p\Phi(\ket{0}_{q_1}\bra{0}) \rangle\\
     &\to \langle E,~~~~(1-p)^2\ket{0}_{q_1}\bra{0} + p(1-p)H\Phi(\ket{0}_{q_1}\bra{0})H + p \Phi(\rho_1) \rangle
  \end{align*}
  where $\rho_1 = (1-p)\ket{+}_{q_1}\bra{+} + p\Phi(\ket{0}_{q_1}\bra{0})$.
  Here, the desired output state $\ket{0}_{q_1}\bra{0}$ is in superposition with error terms $H\Phi(\ket{0}_{q_1}\bra{0})H$ and $\Phi(\rho_1)$.
  This means that not all of the photons will necessarily remain on the `0' path.
  Some may end up on the `1' path.

  As an example, consider the case where the error probability $p$ is $0.1$ and the error model is defined by $\Phi(\rho) = X\rho X$.
  This means that with probability $0.1$, an $X$ gate is applied instead of an $H$ gate.
  With this model, the final state of the system will be $0.91\ket{0}_{q_1}\bra{0} + 0.09\ket{1}_{q_1}\bra{1}$.
  So there is a 9\% chance that a photon will end up on the `1' path.
\end{example}

\section{Quantum Robustness}
\label{sec:robustness}

This section defines a notion of \emph{quantum robustness}, which
bounds the distance between the output of a noisy execution
of a program and the ideal (noise-free) execution of
the same program. We first introduce distance measures in quantum information, and then present the semantic definition of
robustness and define a logic for reasoning about it, which we
prove sound. 

\subsection{Distance Measures of Quantum States and Superoperators}
In the context of computation with noise, we wish to measure the distance between the ideal state and the state influenced by noise.
In classical probabilistic computation, the output can be described as a probability distribution over all possible outputs.
A common measure is the \emph{total variation} distance of two distributions $p,q$, defined by $|p-q|=1/2 \sum_x|p(x)-q(x)|$.

In quantum computation, we define the \emph{trace distance} as the quantum generalization of the total variation distance.
For Hermitian operator $A$, let the trace norm $\nm{A}_1$ be the summation of the absolute value of all its eigenvalues. 
When $A$ is positive semidefinite, one has $\nm{A}_1=\tr(A)$.
It is worth noting that for any operators $A$ and $B$ the following triangle inequality holds:
\begin{equation}
   \nm{A + B}_1 \leq \nm{A}_1 + \nm{B}_1.
\end{equation}
For two states $\rho_1, \rho_2 \in\D(\H)$, the \emph{trace distance} between $\rho_1$ and $\rho_2$, denoted $\tdist{\rho_1}{\rho_2}$, is defined to be $\frac{1}{2}\nm{\rho_1-\rho_2}_1$.
It can be shown that the trace distance satisfies
\begin{align}\label{eq:tracedist}
\tdist{\rho_1}{\rho_2} \equiv \frac{1}{2} \nm{\rho_1-\rho_2}_1=\max_{0\sqsubseteq P\sqsubseteq I}\tr(P(\rho_1-\rho_2)),
\end{align}
where the maximization is over all possible measurements $P$ (i.e., $0\sqsubseteq P\sqsubseteq I$). By definition, $0\leq \tdist{\rho}{\sigma} \leq 1$ for any pair of states $\rho, \sigma$.
One can also interpret the trace distance as the advantage with which one can distinguish two states $\rho$ and $\sigma$.
For example, the states $\ket{0}\bra{0}$ and $\ket{1}\bra{1}$ are perfectly distinguishable by measurement $\{\ket{0}\bra{0}, \ket{1}\bra{1}\}$, so $\tdist{\ket{0}\bra{0}}{\ket{1}\bra{1}}=1$.

One can further define the distance between two different superoperators. 
For any two superoperators $\E$ and $\E'$ over $\H$, an intuitive way to define their distance is to find the input state to both superoperators that maximizes the trace distance of their output states.
However, it turns out that this definition alone is insufficient to capture the distance between superoperators because of a unique quantum feature called \emph{entanglement}.
The distinguishability between $\E$ and $\E'$ can be significantly enlarged
\footnote{For example, consider a 2-dimensional Hilbert space $\H=\mathbb{C}^2$ and two superoperators $\E,\E'$ over $\H$ with $\E(\rho)=\frac{1}{3}\tr(\rho)I_\H+\frac{1}{3}\rho^T$ and $\E'(\rho)=\tr(\rho)I_\H-\rho^T$, where $\rho^T$ is the transpose of $\rho$.
Without introducing an auxiliary space, direct calculation gives $\tdist{\E(\rho)}{\E'(\rho)}=\frac{1}{3}\nm{I-2\rho}_1\leq\frac{2}{3}$ for any single qubit state $\rho$.
However, if we use an auxiliary qubit and apply $\E$ to the maximally entangled state $\ket{\phi^+}=\frac{1}{\sqrt{2}}(\ket{00}+\ket{11})$, we have $\tdist{\E(\ket{\phi^+}\bra{\phi^+})}{\E'(\ket{\phi^+}\bra{\phi^+})}=1$.}
 when one introduces an auxiliary Hilbert space $\mathcal{A}$ and tries to distinguish $\E\otimes I_{\A}$ and $\E'\otimes I_{\A}$ with some \emph{entangled} input state over $\H \otimes \A$~\citep{diamond-norm}.

To account for this, we define the \emph{diamond} norm between two superoperators $\E,\E'$ as
\begin{align}\label{eq:diamond}
  \nm{\E-\E'}_\diamond \equiv \max_{\rho\in\D(\H\otimes\mathcal{A})~:~\tr(\rho) = 1} \tdist{\E\otimes I_{\mathcal{A}}(\rho)}{\E'\otimes I_{\mathcal{A}}(\rho)}
\end{align}
for any auxiliary space $\mathcal{A}$. Without loss of generality, one can assume $\A$ is a copy of $\H$, i.e., $\A=\H$~\citep{Wat18}.
Given the representations of $\E,\E'$, one can efficiently calculate the diamond norm $\nm{\E-\E'}_\diamond$ by a semidefinite program (SDP)~\citep{W09}.

Imagine superoperators that represent different components of a (noisy or ideal) quantum program, which may act on different parts of the quantum system. 
The diamond norm will allow us to address potential entanglement between different parts of the state and ensure that we can compose the distances computed from different components of the program.

\subsection{Definition of Quantum Robustness}\label{sec:judgment}

To capture how noise (error) impacts the execution of quantum program
$\widetilde{P}$, we want to compare $\sem{\widetilde{P}}$  and
$\sem{P}$, which are superoperators representing the execution of
quantum program $P$ with and without noise respectively. A natural
candidate is to use the aforementioned diamond norm to measure the
distance between $\sem{\widetilde{P}}$  and $\sem{P}$.
To account for prior knowledge of the input state, we extend the definition of the diamond norm to consider only input states that satisfy predicate $Q$ to degree at least $\lambda$. More explicitly, we have

\begin{definition}[$(Q,\lambda)$-diamond norm] \label{def:q-diamond}
Given superoperators $\E$, $\E'$, quantum predicate $Q$ over $\H$, and $0\leq \lambda \leq 1$, the $(Q, \lambda)$-diamond norm between $\E$ and $\E'$, denoted $\nm{\E-\E'}_{Q,\lambda}$,  is defined by
\begin{align}\label{eq:diamond-q}
  \nm{\E-\E'}_{Q,\lambda} \equiv \max_{\rho\in\D(\H\otimes\mathcal{A})~:~\tr(\rho) = 1,~\tr(Q\rho)\geq\lambda} \tdist{\E\otimes I_{\mathcal{A}}(\rho)}{\E'\otimes I_{\mathcal{A}}(\rho)},
\end{align}
where $\A$ is any auxiliary space. 
\end{definition}
We remark that $\A$ can be assumed to be $\H$
without loss of generality due to a similar reason for the original diamond norm (see, for example,~\citet[Chap 3]{Wat18}). 

We argue that $(Q,\lambda)$-diamond norm is a seminorm in \app{Qlambda}. 
Intuitively, this is conceivable since we only restrict the input state from all density operators to a convex subset satisfying $\tr(Q\rho)\geq \lambda$.  Note that, by definition, $\nm{\cdot}_\diamond \equiv \nm{\cdot}_{I, \lambda}$ for any $0\leq \lambda \leq 1$. 

Let $\E$ and $\E'$ be superoperators over $\H$. By extending~\citet{W09}, we show that $\nm{\E-\E'}_{Q,\lambda}$ can be efficiently computed by the following semidefinite program (SDP):
\begin{align}
 \max &  \quad \tr(J(\Phi) W) \\
 \mathrm{s.t.} & \quad  W \leq I_\H \otimes \rho,~\tr(Q\rho) \geq \lambda, \\
    & \quad  \rho \in \D(\H),~W \text{ is a positive semidefinite operator over } \H \otimes \H,
\end{align} 
where $\Phi=\E-\E'$ and $J(\Phi)$ is the \emph{Choi-Jamiolkowski} representation of $\Phi$. 
Note that the above SDP is identical to the SDP used in \citet[Section 4]{W09} to compute $\nm{\E-\E'}_\diamond$, except that we have added the additional constraint $\tr(Q\rho) \geq \lambda$
to capture the requirement on input states. The correctness of the above SDP then basically follows from the analysis of \citet{W09} and the definition of $(Q, \lambda)$-diamond norm. 

% \begin{remark}
The standard diamond norm and $(Q,\lambda)$-diamond norm can be significantly different for the same pair of superoperators $\E, \E'$. For example, consider $\E=H\circ H$ and $\E'=HZ\circ ZH$. We can show 
\footnote{For the normal diamond norm, consider the input state to be $\rho=\ket{+}\bra{+}\otimes\sigma$ for some ancilla state $\sigma$. It is easy to see that $\E(\rho)=\ket{0}$ and $\E'(\rho)=\ket{1}$, which are perfectly distinguishable. For the $(\ket{0}\bra{0}, \frac{3}{4})$-diamond norm, without loss of generality, consider any input state $\ket{\psi}=\cos\theta\ket{0}\ket{\psi_0}+\sin\theta\ket{1}\ket{\psi_1}$ where $\theta\in[0,\frac{\pi}{2}]$.
Requiring $\ket{\psi}$ to satisfy $\ket{0}\bra{0}$ to degree at least $\frac{3}{4}$, we have $\cos^2\theta\geq \frac{3}{4}$ and therefore $\theta\in[0,\frac{\pi}{6}]$.
Simple calculation gives
$H\ket{\psi}=\cos\theta\ket{+}\ket{\psi_0}+\sin\theta\ket{-}\ket{\psi_1}$
and $HZ\ket{\psi}=\cos\theta\ket{+}\ket{\psi_0}-\sin\theta\ket{-}\ket{\psi_1}$.
The projector that maximally distinguishes these states is $\ket{0}\bra{0}\otimes I$, so we have
\begin{align*}
  \nm{\E-\E'}_{\ket{0}\bra{0},3/4}
  =\frac{1}{2} \left ((\cos\theta+\sin\theta)^2 -
  (\cos\theta-\sin\theta)^2 \right)
  =\sin 2\theta\leq \frac{\sqrt{3}}{2},
\end{align*}
the equality of which holds when $\theta=\pi/6$.
}
that $\nm{\E-\E'}_\diamond=1$, whereas $\nm{\E-\E'}_{\ket{0}\bra{0},\frac{3}{4}}=\sqrt{3}/2$. 
Thus, the $(Q,\lambda)$-diamond norm can help us leverage prior knowledge about input states to obtain more accurate bounds.
% \end{remark}

Using the $(Q, \lambda)$-diamond norm, we define a notion of quantum
robustness as follows.

\begin{definition}[Quantum Robustness] \label{def:r_judgment}
The noisy program $\widetilde{P}$ (over $\H$, having ideal program 
  $P=\ideal(\widetilde{P})$) is \emph{$\epsilon$-robust under
    $(Q,\lambda)$} if and only if 
\begin{equation} 
  \label{eq:qrobust}
  \nm{\sem{\widetilde{P}}- \sem{P}}_{Q, \lambda} \leq \epsilon. 
\end{equation}
Here, $Q$ is a quantum predicate over $\H$ and $0 \leq \lambda, \epsilon \leq 1$.  

By the definition of the $(Q, \lambda)$-diamond norm and
the trace distance, \eq{qrobust} can be equivalently stated as the
following. 
\begin{align}
  \label{eq:analysis-judgment}
  \forall\rho \in \D(\H \otimes \H),~\tr(Q\rho)\geq\lambda\tr(\rho)\Rightarrow \frac{1}{2}\nm{\sem{\widetilde{P}}\otimes I_{\H} (\rho)-\sem{P}\otimes I_{\H} (\rho)}_1\leq \epsilon\tr(\rho).
\end{align}
\end{definition}

Since $\epsilon$ measures the distance between $\sem{\widetilde{P}}$ and
$\sem{P}$, the \emph{smaller} $\epsilon$ is, the \emph{closer} the noisy program $\widetilde{P}$
is to the ideal program $P$. 
One could think of $\epsilon$ as measuring both the probability that
noise can happen and intensity of that noise. 
When the noise is strong, a noisy program $\widetilde{P}$ being $\epsilon$-robust implies that the probability of noise is at most $\epsilon$. 
When the noise is weak, it could occur with greater probability, but
its effect will be much smaller. 

Intuitively, the use of precondition $Q$ can help us obtain more accurate bounds. For example, one can use $Q$ to characterize prior information about the input state due to the nature of underlying physical systems. 
Even without any prior knowledge about the input state, preconditions can still be leveraged for different branches of the program in case statements and loops. 

\subsection{Logic for Quantum Robustness}\label{sec:noisy-rules}

A program's robustness can be proved by working out the
(denotational) semantics of programs $P$ and $\widetilde{P}$ and
applying Definition~\ref{def:r_judgment} directly. However, this computation
may be difficult. As an alternative, we present a logic for proving judgments of the
form $\rjudgment{Q}{\lambda}{\widetilde{P}}{\epsilon}$, meaning that
program $\widetilde{P}$ is $\epsilon$-robust under
$(Q,\lambda)$. In \sec{soundness} we prove the logic is sound. In
\sec{case-studies} we demonstrate the use of the logic, alongside
direct proofs of robustness for some cases (e.g., error correction).

\begin{figure}
  \begin{gather*}
    \infer[\text{(Skip)}]{\rjudgment{Q}{\lambda}{\cskip}{0}}{}
    \qquad
    \infer[\text{(Init)}]{\rjudgment{Q}{\lambda}{(q:=\ket{0})}{0}}{}\\
    \infer[\text{(Unitary)}]
      {\rjudgment{Q}{\lambda}{(\overline{q}:\cong_{p,\Phi}U[\overline{q}])}{p\epsilon}}
      {\nm{U\circ U^\dag-\Phi}_{Q,\lambda}\leq\epsilon}\\
    \infer[\text{(Weaken)}]
      {\rjudgment{Q}{\lambda}{\widetilde{P}}{\epsilon}}
      {\rjudgment{Q'}{\lambda'}{\widetilde{P}}{\epsilon'} 
       \qquad \epsilon'\leq\epsilon 
       \qquad Q\sqsubseteq Q'
       \qquad \lambda'\leq\lambda}\\
    \infer[\text{(Rescale)}]
      {\rjudgment{Q}{\lambda}{\widetilde{P}}{\epsilon}}
      {\rjudgment{Q/\delta}{\lambda/\delta}{\widetilde{P}}{\epsilon} 
       \qquad {0\sqsubseteq Q, Q/\delta\sqsubseteq I 
       \qquad 0 \leq \lambda, \lambda/\delta \leq 1}} \\
    \infer[\text{(Sequence)}]
      {\rjudgment{Q_1}{\lambda}{(\widetilde{P_1};\widetilde{P_2})}{\epsilon_1+\epsilon_2}}
      {\rjudgment{Q_1}{\lambda}{\widetilde{P_1}}{\epsilon_1} 
       \qquad \rjudgment{Q_2}{\lambda}{\widetilde{P_2}}{\epsilon_2} 
       \qquad \{Q_1\}P_1\{Q_2\}}\\
    \infer[\text{(Case)}]
      {\rjudgment{\sum_m M_m^\dag Q_m M_m}
                  {1-t\delta}
                  {(\qif{M[\overline{q}]=\overline{m\to \widetilde{P_m}}} )}
                  {(1-t)\epsilon+t}}
      {\forall m, \rjudgment{Q_m}{1-\delta}{\widetilde{P_m}}{\epsilon}
       \qquad t,\delta\in[0,1]}\\
    \infer[\text{(While-Bounded)}]
      {\rjudgment{\lambda M_0^\dag M_0+M_1^\dag QM_1}
                  {\lambda}
                  {\widetilde{P}}
                  {n\epsilon/(1-a)}}
      {\begin{array}{c}
         \rjudgment{Q}{\lambda}{\widetilde{P_1}}{\epsilon}
         \qquad \{Q\}P_1\{\lambda M_0^\dag M_0+ M_1^\dag Q M_1\} \\
         \widetilde{P}\equiv\qwhile{M[\bar{q}]=1}{\widetilde{P_1}}
         \qquad P\text{ is } (a,n)\text{-bounded}
       \end{array}}\\
    \infer[\text{(While-Unbounded)}]
      {\rjudgment{Q}{\lambda}{(\qwhile{M[\bar{q}]=1}{\widetilde{P_1}})}{1}}
      {}
  \end{gather*}
  \caption{Rules for logic of quantum robustness.}
  \label{fig:reliability-sem}
\end{figure}

The rules for our logic are given in \fig{reliability-sem}.
\subsubsection{Simple Rules} 
The \emph{Skip} and \emph{Init} rules say that the skip and initialization operations are always error-free.
These operations will not increase the distance between $\sem{P}$ and $\sem{\widetilde{P}}$.
The \emph{Unitary} rule says that if we can bound the $(Q, \lambda)$-diamond norm between the intended operation $U$ and the noise operation $\Phi$ by $\epsilon$, then we can bound the total distance by $p\epsilon$.
The \emph{Weaken} rule says that we can always safely make the precondition more restrictive, increase the degree to which an input state must satisfy the predicate, or increase the upper bound on the distance between the noisy and ideal programs.
The \emph{Rescale} rule says that equivalent forms of our judgment can be obtained by rescaling $Q$ and $\lambda$. Note that the \emph{Rescale} rule does not weaken the judgment, but rather provides some flexibility in choosing $Q,\lambda$ compatible with other rules; one can scale by $\delta$ so long as $Q/\delta$ and $\lambda/\delta$ are still well defined.  
The \emph{Sequence} rule allows us to compose two judgments by summing their computed upper bounds.
Note that in the \emph{Sequence} and \emph{While-Bounded} rules we define Hoare triples as in \sec{qpredicate}.

\subsubsection{The Case Rule}
The \emph{Case} rule says that, given appropriate bounds for every branch of a case statement, we can bound the error of the entire case statement.
%%% XW: this is a technical point that cannot be easily spoken out in the introduction. 
Note that $\sum_m M_m^ \dag Q_m M_m$ is the weakest precondition of the case construct in quantum Hoare logic~\citep{Yin16}. %, a favorable property for the rule. 
%% the following is also hard to speak in the intro.... too technical; 
In a logic for classical programs, one might expect each branch of a case statement to satisfy the precondition perfectly.
However, in a quantum logic, as we will discuss in the soundess proof (see \sec{soundness}), this is not necessarily true.
To see this, note that in our rule we start with the precondition $\sum_m M_m^ \dag Q_m M_m$ and $\lambda=1- t \delta$ on the input state to the case statement, but we can only guarantee that a weighted fraction of $1-t$ of the branches satisfy a weaker precondition $Q_m$ and $\lambda'=1-\delta$ for some choice of $t \in [0,1]$. 
(Note that $1-\delta \leq 1-t\delta$ for $t, \delta \in [0,1]$.) 

When applying this rule, one can make $1-\delta$ and $\epsilon$ the same for every $\widetilde{P_m}$ by applying the \emph{Weaken} and/or \emph{Rescale} rules. 
The choice of $t$ will represent a tradeoff between a lower error bound and a more restrictive requirement on the satisfaction of the predicate (see \ex{t}).

\begin{example}[Simple case statement]
  Consider the following program.
  \begin{alignat*}{2}
    \widetilde{P} \equiv~ & \qif{M[\overline{q}] =~& &0 \to \overline{q}:\cong_{p,\Phi} H[\overline{q}]\\
                          &                        & &1 \to \cskip\\
                          & }                      & &
  \end{alignat*}
  This program performs measurement in the $\{ \ket{0}, \ket{1} \}$ basis and either applies a noisy Hadamard $H$ gate or does nothing depending on the measurement outcome.
  In order to apply the \emph{Case} rule, we must bound the error of both branches.
  By the \emph{Skip} rule, the second branch will have zero error, i.e., $\rjudgment{Q}{\lambda}{\cskip}{0}$ for any $Q$ and $\lambda$.
  We will choose $Q = I$ and $\lambda = 1$.
  By the \emph{Unitary} rule, the error of the first branch will depend on $\Phi$.

  Consider an error model given by $\Phi(\rho) = HZ\rho ZH$.
  This means that with probability $1-p$, $\overline{q}:\cong_{p,\Phi} H[\overline{q}]$ applies an $H$ gate, and with probability $p$ it applies a $Z$ gate followed by an $H$ gate. %, where $Z$ is the Pauli $Z$ matrix.
  Recall that $Z\ket{0}\bra{0}Z = \ket{0}\bra{0}$.
  This means that a $Z$ error will not affect the state in the first branch (because the first branch corresponds to having measured a 0).
  So we have that $\rjudgment{\ket{0}\bra{0}}{1}{(\overline{q}:\cong_{p,\Phi} H[\overline{q}])}{0}$, which says that if the input state is the $\ket{0}$ state, then with the error model defined, there will be no error during an application of $H$ (i.e. $\epsilon = 0$).
  Note that without the precondition $\ket{0}\bra{0}$, we would have $\epsilon > 0$.

  Given $\rjudgment{\ket{0}\bra{0}}{1}{(\overline{q}:\cong_{p,\Phi} H[\overline{q}])}{0}$ and $\rjudgment{I}{1}{\cskip}{0}$, we can conclude that
  $\rjudgment{I}{1}{\widetilde{P}}{t}$ by the \emph{Case} rule with the following simplification, 
\[
    M_0^\dag Q_0 M_0 + M_1^\dag Q_1 M_1 = \ket{0}\bra{0}\ket{0}\bra{0}\ket{0}\bra{0} + \ket{1}\bra{1}I\ket{1}\bra{1}= \ket{0}\bra{0} + \ket{1}\bra{1} = I.
\]
  Because $\delta  = 0$, we can choose $t = 0$ to get an overall error bound of zero.
  However, in general, the choice of $t$ will represent a tradeoff between a lower error bound and a more restrictive requirement on the satisfaction of the predicate (i.e. a larger $\lambda$ value).
\end{example}

%\begin{remark}\normalfont
  Note that the precondition for each branch does not need to directly relate to the measurement basis.
  For example, consider instead the following program.
  \begin{alignat*}{2}
    & \qif{M[\overline{q}] =~& &+ \to \overline{q}:\cong_{p,\Phi} H[\overline{q}]\\
    &                        & &- \to \cskip\\
    & }                      & &
  \end{alignat*}
  This program performs measurement in the $\{ \ket{+}, \ket{-}\}$ basis.
  For this program, we can still use the judgments $\rjudgment{\ket{0}\bra{0}}{1}{(\overline{q}:\cong_{p,\Phi} H[\overline{q}])}{0}$ and $\rjudgment{I}{1}{\cskip}{0}$ described in the previous example,
  but now our conclusion will be $\rjudgment{\frac{1}{2}\ket{+}\bra{+} + \ket{-}\bra{-}}{1}{\widetilde{P}}{t}$, by the following simplification, 
 \[
  M_0^\dag Q_0 M_0 + M_1^\dag Q_1 M_1 = \ket{+}\bra{+}\ket{0}\bra{0}\ket{+}\bra{+} + \ket{-}\bra{-}I\ket{-}\bra{-}= \frac{1}{2}\ket{+}\bra{+} + \ket{-}\bra{-}.
  \]
  Note that $\frac{1}{2}\ket{+}\bra{+} + \ket{-}\bra{-} \sqsubseteq I$, so we have restricted the set of states for which this error bound will hold.
  This may be useful in situations where it is difficult to compute an error bound given certain preconditions.
  For example, it may be difficult to show that $\rjudgment{\ket{+}\bra{+}}{1}{(\overline{q}:\cong_{p,\Phi} H[\overline{q}])}{\epsilon}$ for a sufficiently low $\epsilon$, but it is easy to show that $\rjudgment{\ket{0}\bra{0}}{1}{(\overline{q}:\cong_{p,\Phi} H[\overline{q}])}{0}$.
%\end{remark}

\begin{example}[t-value tradeoffs] \label{ex:t}
  Consider the following program.
  \begin{alignat*}{2}
    \widetilde{P} \equiv~ & \qif{M[\overline{q}] =~& &0 \to \widetilde{P_0}\\
                          &                        & &1 \to \widetilde{P_1}\\
                          & }                      & &
  \end{alignat*}
  Say that we have that $\rjudgment{Q_0}{1 - 0.05}{\widetilde{P_0}}{0.01}$ and $\rjudgment{Q_1}{1 - 0.25}{\widetilde{P_1}}{0.1}$ for programs $\widetilde{P_0}$ and $\widetilde{P_1}$.
  Now we can use the \emph{Case} rule to conclude that 
  $\rjudgment{M_0^\dag Q_0 M_0 + M_1^\dag Q_1 M_1}{1 - t\delta}{\widetilde{P}}{(1-t)\epsilon+t}$
  for $\delta = \min(0.05, 0.25) = 0.05$, $\epsilon = \max(0.01, 0.1) = 0.1$, and $t \in [0,1]$.

  If we choose $t=0$ then we have the lowest possible error bound, but $1 - t\delta = 1$, so the produced error bound will only apply to states that completely satisfy the predicate $M_0^\dag Q_0 M_0 + M_1^\dag Q_1 M_1$.\footnote{This might even be impossible when $M_0^\dag Q_0 M_0 + M_1^\dag Q_1 M_1\sqsubset I$. In that case, the choice $t=0$ becomes useless.}
  If we choose $t=1$ then we have the least restrictive condition on input states, but $(1-t)\epsilon + t = 1$, which is the trivial error bound.
  By choosing $t$ to be between 0 and 1 we can trade between a low error bound and a less restrictive constraint on the input state.
  For example, for $t=0.25$ we have that $1 - t\delta  = 0.9875$ and $(1-t)\epsilon + t = 0.325$.
  For $t=0.5$ we have that $1 - t\delta = 0.975$ and $(1-t)\epsilon + t = 0.55$.
\end{example}

\subsubsection{The Loop Rule}
Finally, the \emph{While-Bounded} and \emph{While-Unbounded} rules allow us to compute an upper bound for the distance between the noisy and ideal versions of a while loop.
The \emph{While-Unbounded} rule is a trivial bound on the distance. The \emph{While-Bounded} rule, however, demonstrates a non-trivial upper bound with an assumption called $(a,n)$-boundedness. 
Such an assumption is necessary for us to get around the potential issue of termination and to be able to reason about interesting programs in our case study. 

Intuitively, $(a,n)$-boundedness is a condition on how fast the \emph{ideal} loop will converge, which is inherent to the control flow of the program and does not depend on any specific error model. We view this as an advantage as we do not need a new analysis for every possible noise model. 
%\xw{an advantage; don't need to redo the analysis}
\begin{definition}[$(a,n)$-boundedness] \label{def:an}
A while loop $P\equiv\qwhile{M[q]=1}{P_1}$ is said to be $(a,n)$-bounded for $0\leq a<1$ and integer $n\geq 1$ if
\begin{align}\label{eq:anbounded}
  (\E^*)^n(M_1^\dag M_1)\sqsubseteq aM_1^\dag M_1
\end{align}
where the linear map $\E(\rho)$ is defined as $\sem{P_1}(M_1\rho M_1^\dag)$ and $\E^*$ is the dual map of $\E$.\footnote{If $\sem{P_1}$ can be written as $\sum_k F_k\circ F_k^\dag$ for some set of Kraus operators $\{F_k\}_k$,
the Kraus form of $\E^*$ is $\sum_k M_1^\dag F_k^\dag\circ F_k M_1$.}
\end{definition}

Intuitively speaking, a loop is $(a,n)$-bounded if, for every state $\rho$, after $n$ iterations it is guaranteed that at least a $(1-a)$-fraction of the state has exited the loop.
A while loop with this nice property is guaranteed to terminate with probability $1$ on all input states, which helps avoid the termination issue. As we will show in the examples that follow and in \sec{case-studies}, specific $(a,n)$ can be derived analytically or numerically for concrete programs.\footnote{The rough idea is to guess (or enumerate) $n$ and prove $a$ either analytically or numerically (with a simple SDP).} 
We also remark that by assuming $(a,n)$-boundedness, we avoid weakening $\lambda$ like in the \emph{Case} rule. 

In~\citet{CMR2013}, loops are assumed to have a bounded number of iterations or a trivial upper bound will be used (i.e, our rule \emph{While-Unbounded}). Because of the use of $(a,n)$-boundedness, we can handle more complicated loops. One also has the freedom to choose appropriate values of $Q$ for different purposes. A simple choice is $Q=\lambda I$. A less trivial choice of $Q$ is shown in the following example. 

\begin{example}[Slow state preparation]\label{ex:ssp}
  In this example, we consider the following program which prepares the standard basis state $\ket{1}$:
  \begin{align}
    SSP\equiv q:=\ket{0};\qwhile{M[q]=0}{q:=H[q]; q:=I[q]},
  \end{align}
  where $M$ is the standard basis measurement $\{ \ket{0}\bra{0}, \ket{1}\bra{1}\}$.  
  We consider the case where there is a bit flip error with probability 0.01 when applying the $I$ gate, i.e.,
  the ideal and the noisy loop bodies are $P_1\equiv q:=H[q]; q:=I[q];$ and $\widetilde{P_1}\equiv q:=H[q]; q:\cong_{0.01,X}I[q]$.
  Then $\sem{P_1}=H\circ H$ and $\sem{\widetilde{P_1}}=0.99H\circ H+0.01 XH\circ HX$.
  Consider $Q = \ket{0}\bra{0}$ and $\lambda = 1$.
  Since $\sem{P_1}(\ket{0}\bra{0})=\sem{\widetilde{P_1}}(\ket{0}\bra{0})=\ket{+}\bra{+}$, we have that $\nm{\sem{P_1}-\sem{\widetilde{P_1}}}_{\ket{0}\bra{0},1}=0$, and by the \emph{Unitary} rule, $\rjudgment{\ket{0}\bra{0}}{1}{\widetilde{P_1}}{0}$.
  
  We can use this choice of $Q$ and $\lambda$ when applying the \emph{While-Bounded} rule.
  First, note that the statement $\{Q\}P_1\{\lambda M_0^\dag M_0 + M_1QM_1^\dag\}$ holds because $\lambda M_0^\dag M_0 + M_1QM_1^\dag = I$ for our choice of $M_0, M_1, Q, \lambda$. 
  Next, we need to show $(a,n)$-boundedness of the loop. 
  This requires us to consider the behavior of $\E^*$ where $\E^*$ is the dual of $\E=HM_1\circ M_1H$.   
  We claim that the while loop is $(1/2,1)$-bounded because 
  \[
   \E^*(M_1^\dag M_1)=\ket{0}\bra{0}H\ket{0}\bra{0}H\ket{0}\bra{0}=\frac{1}{2}\ket{0}\bra{0}=\frac{1}{2} M_1^\dag M_1. 
   \]
  Now by the \emph{While-Bounded} rule, we have that $\rjudgment{I}{1}{\widetilde{SSP}}{0}$, i.e.,
  the program is perfectly robust. % in the presence of noise.
\end{example}
%\begin{remark}
  In \ex{ssp}, if we use the precondition $I$ in our judgment of the robustness of the loop body,
  the best upper bound we can argue is $\epsilon=0.01$, i.e., $\rjudgment{I}{1}{\widetilde{P_1}}{0.01}$. 
  Then, applying the \emph{While-Bounded} rule yields $\rjudgment{I}{1}{\widetilde{SSP}}{0.02}$ since $\frac{n\epsilon}{1-a}=2\epsilon=0.02$.
  Therefore, restricting the state space with the predicate $Q=\ket{0}\bra{0}$, as we did in \ex{ssp}, was a better choice, as it yielded a better bound.
  Moreover, this choice of $Q$ is natural 
  since the post-measurement state entering the loop satisfies $Q$ perfectly.
%\end{remark}
%\begin{remark}
  We note that the program $SSP$ might look contrived for preparing $\ket{1}$ when compared to the more straightforward program $SP\equiv q:=\ket{0}; q:=X[q]$. 
  We argue that $SSP$ might be preferred over $SP$ in the presence of noise. Consider the same noise model as above. Namely, let $\widetilde{SP}\equiv q:=\ket{0}; q:=X[q]; q:\cong_{0.01, X}I[q].$
  We have shown that $\rjudgment{I}{1}{\widetilde{SSP}}{0}$ given this noise model, which says that $\widetilde{SSP}$ is perfectly robust.
  By directly applying Definition~\ref{def:r_judgment}, we can show that $\widetilde{SP}$ is $0.01$-robust given the same noise model,\footnote{
  Note that $\sem{\widetilde{SP}}=0.99 X\circ X + 0.01 I\circ I$, and hence we have $\nm{\sem{SP}-\sem{\widetilde{SP}}}_{I,1} = 0.01\nm{X\circ X-I\circ I}_{I,1}=0.01$.} which suggests that $\widetilde{SP}$ is less desirable in this situation. 
%\end{remark}

\subsection{Soundness} \label{sec:soundness} In this section, we show that logic given in \fig{reliability-sem} is sound.

\begin{theorem}[Soundness]
If $\rjudgment{Q}{\lambda}{\widetilde{P}}{\epsilon}$ then  $\widetilde{P}$ is $\epsilon$-robust under $(Q,\lambda)$.
\end{theorem}
\begin{proof}
The proof proceeds by induction on the derivation $\rjudgment{Q}{\lambda}{\widetilde{P}}{\epsilon}$.
We will work primarily from definition of robustness given in \eq{analysis-judgment}.
We note that superoperators $\sem{\widetilde{P}}$ and $\sem{P}$ apply on the same space. By the definition of the $(Q, \lambda)$-diamond norm, we need to consider 
$\sem{\widetilde{P}} \otimes I$ and $\sem{P}\otimes I$. However, to simplify the presentation, we will omit ``$\otimes I$'' in the following proof whenever there is no ambiguity. 

Now we consider each possible rule used in the final step of the derivation.
\begin{enumerate}
\item \emph{Skip:} This rule holds by observing $\sem{\widetilde{\cskip}}=\sem{\cskip}$, i.e., they refer to the same superoperator.  Thus, any $(Q, \lambda)$-diamond norm between them is 0. We choose $Q=I$ and $\lambda=0$. 
\item \emph{Init:} for the same reason as the proof for \emph{Skip}. 
\item \emph{Unitary:} For every state $\rho$ satisfying $\tr(Q\rho)\geq\lambda$,
  \begin{eqnarray*}
    \frac{1}{2}\nm{\sem{\bar{q}:\cong_{p,\Phi}U[\bar{q}]}\rho-\sem{\bar{q}:=U[\bar{q}]}\rho}_1
      &=& \frac{1}{2}\nm{\big((1-p)U\rho U^\dag + p\Phi(\rho)\big) - U\rho U^\dag}_1 \\
    &=& \frac{p}{2}\nm{U\rho U^\dag-\Phi(\rho)}_1 \leq p\nm{U\cdot U^\dag-\Phi}_{Q,\lambda} \leq p\epsilon.
  \end{eqnarray*}
The second from last inequality holds by the definition of the $(Q,\lambda)$-diamond norm and the last inequality follows from the premise. 
\item \emph{Weaken:} By induction, the premise $\rjudgment{Q'}{\lambda'}{\widetilde{P}}{\epsilon'}$ implies
  \[\forall \rho, \tr(Q'\rho)\geq\lambda'\tr(\rho)\Rightarrow \frac{1}{2}\nm{\sem{\widetilde{P}}\rho-\sem{P}\rho}_1\leq \epsilon'\tr(\rho).\]
  For any density matrix $\rho$, constants $0 \leq \lambda' \leq \lambda \leq 1$, and predicates $Q \sqsubseteq Q'$, $\tr(Q\rho)\geq\lambda\tr(\rho)$ implies that $\tr(Q'\rho)\geq\lambda'\tr(\rho)$.
  And for $0 \leq \epsilon' \leq \epsilon \leq 1$, $\frac{1}{2}\nm{\sem{\widetilde{P}}\rho-\sem{P}\rho}_1\leq \epsilon'\tr(\rho)$ implies that $\frac{1}{2}\nm{\sem{\widetilde{P}}\rho-\sem{P}\rho}_1\leq \epsilon\tr(\rho)$.
  Therefore, we have that \[\forall \rho, \tr(Q\rho)\geq\lambda\tr(\rho)\Rightarrow \frac{1}{2}\nm{\sem{\widetilde{P}}\rho-\sem{P}\rho}_1\leq \epsilon\tr(\rho).\]
  So $\widetilde{P}$ is $\epsilon$-robust under $(Q,\lambda)$. 
\item \emph{Rescale:} This rule follows by observing that the condition $\tr(\rho Q) \geq \lambda$ is equivalent to the condition $\tr(\rho Q/\delta) \geq \lambda/\delta$ for $\delta>0$. We only require that $Q, Q/\delta$ and $\lambda, \lambda/\delta$ are well defined, namely, $0\sqsubseteq Q, Q/\delta \sqsubseteq I$ and $0 \leq \lambda, \lambda/\delta \leq 1$.
\item \emph{Sequence:} For every state $\rho$,
  \begin{eqnarray*}
    \nm{\sem{\widetilde{P_1};\widetilde{P_2}}\rho - \sem{P_1;P_2}\rho}_1 &=& \nm{\sem{\widetilde{P_2}}\sem{\widetilde{P_1}}\rho-\sem{P_2}\sem{P_1}\rho}_1 \\
    &\leq& \nm{\sem{\widetilde{P_2}}\sem{\widetilde{P_1}}\rho-\sem{\widetilde{P_2}}\sem{P_1}\rho}_1 +
    \nm{\sem{\widetilde{P_2}}\sem{P_1}\rho-\sem{P_2}\sem{P_1}\rho}_1 \\
    &\leq& \nm{\sem{\widetilde{P_1}}\rho-\sem{P_1}\rho}_1
    +\nm{\sem{\widetilde{P_2}}\sem{P_1}\rho-\sem{P_2}\sem{P_1}\rho}_1.
  \end{eqnarray*}
  The inequality $\nm{\sem{\widetilde{P_2}}\sem{\widetilde{P_1}}\rho-\sem{\widetilde{P_2}}\sem{P_1}\rho}_1\leq \nm{\sem{\widetilde{P_1}}\rho-\sem{P_1}\rho}_1$ follows because quantum superoperators are contractive. 
  
  Now assume that $\tr(Q_1\rho) \geq \lambda \tr(\rho)$.
  By induction, the premise $\rjudgment{Q_1}{\lambda}{\widetilde{P_1}}{\epsilon_1}$ implies that $\frac{1}{2}\nm{\sem{\widetilde{P_1}}\rho-\sem{P_1}\rho}_1 \leq \epsilon_1\tr(\rho)$.
  Also, by the premise $\{Q_1\}P_1\{Q_2\}$, we have that $\tr(Q_2\sem{P_1}\rho) \geq \tr(Q_1\rho) \geq \lambda \tr(\rho) \geq \lambda \tr(\sem{P_1}\rho)$.
  Now we can use our induction hypothesis and the premise $\rjudgment{Q_2}{\lambda}{\widetilde{P_2}}{\epsilon_2}$ to conclude $\frac{1}{2}\nm{\sem{\widetilde{P_2}}\sem{P_1}\rho-\sem{P_2}\sem{P_1}\rho}_1 \leq \epsilon_2\tr(\sem{P_1}\rho)$.
  So, finally, we have that
  \[ \frac{1}{2}\nm{\sem{\widetilde{P_1};\widetilde{P_2}}\rho - \sem{P_1;P_2}\rho}_1 \leq \epsilon_1\tr(\rho)+\epsilon_2\tr(\sem{P_1}\rho) \leq (\epsilon_1+\epsilon_2)\tr(\rho). \]
  
\item \emph{Case:} Let $\widetilde{P} = \qif{M[\overline{q}]=\overline{m\to \widetilde{P_m}}}$. 
Assume the input state $\rho$ to the case statement satisfies $\sum_m M_m^\dag Q_m M_m$ to degree $\lambda'$, i.e., $\sum_m \tr(M_m^\dag Q_mM_m\rho)\geq \lambda' \tr(\rho)$.
To leverage the premise $\rjudgment{Q_m}{\lambda}{\widetilde{P_m}}{\epsilon}$ and the induction hypothesis to conclude that $\frac{1}{2}\nm{\sem{\widetilde{P_m}}\rho - \sem{P_m}\rho}_1 \leq \epsilon\tr(\rho)$, one must show the precondition $(Q_m , \lambda)$ holds for state $\rho$ entering branch $m$. 
%  Also let the upper bound of the error parameter for the conditional be $\epsilon'$ if the state satisfies $\sum_m M_m^\dag Q_m M_m$ of degree $\lambda'$.
A naive approach is to show that $\tr(M_m^\dag Q_mM_m\rho)\geq\lambda\tr(M_m^\dag M_m\rho)$ holds for \emph{every} branch $m$, which implies that
$\tr(Q_mM_m\rho M_m^\dag )\geq\lambda\tr( M_m\rho M_m^\dag)$ where $M_m\rho M_m^\dag$ is the (sub-normalized) post-measurement state entering branch $m$. 
We argue that this is in general impossible when $\lambda'=\lambda$. 
For instance, consider a collection of projective measurement operators $\{M_m\}_m$. If there exists a branch $i$ and a state $\rho$ supported on $M_i$ such that $\tr(M_i^\dag Q_i M_i\rho)\geq\lambda\tr(\rho)$ for some $\lambda>0$, obviously $\sum_{m}M_m^\dag Q_m M_m\rho\geq\lambda\tr(\rho)$, but none of the preconditions is satisfied except for the one for branch $i$ since $\tr(M_j^\dag Q_j M_j\rho)=0$ for each $j\neq i$.

Instead, we show that for a majority of the clauses $\tr(M_m^\dag Q_mM_m\rho)\geq\lambda\tr(M_m^\dag M_m\rho)$ holds for some $\lambda$ strictly less than $\lambda'$. 
To that end, let $p_m=\tr(M_m^\dag M_m\rho)$ and $q_m=\tr(M_m^\dag Q_mM_m\rho)$.
Define $\delta_m$ to be such that $q_m=(1-\delta_m)p_m$. Note that $0\leq\delta_m\leq 1$ because $0\leq q_m \leq p_m$ for every $m$. 
  Without loss of generality we assume $\tr(\rho)=1$.
Let $S(\rho)$ denote the collection of branches such that the precondition $(Q_m, \lambda)$ holds, i.e., $S(\rho)=\{m: q_m\geq\lambda p_m, \text{ i.e., } \delta_m \leq 1-\lambda\}$. 
We will determine a lower bound for $\sum_{m\in S(\rho)}p_m$ for each state $\rho$ using a probabilistic argument. 

First, note that $\{p_m\}$ is a probability distribution since $\sum_m p_m=1$ and $p_m\geq 0$ for each $m$.
Also, since (by our assumption) $\sum_m q_m\geq\lambda'\sum_m p_m=\lambda'$, we have that $\sum_m\delta_m p_m\leq (1-\lambda')$.
  Now we can define a random variable $\Delta$ to be such that $\Pr[\Delta=\delta_m]=p_m$ for each $m$.
  The expected value of $\Delta$ is $\mathbb{E}[\Delta]\leq(1-\lambda')$, and Markov's inequality yields
  \begin{align}
    \sum_{m\notin S(\rho)} p_m = \Pr[\Delta\geq 1-\lambda]\leq  \frac{\mathbb{E}[\Delta]}{1-\lambda}  \leq \frac{1-\lambda'}{1-\lambda}=:t.
  \end{align}
  This says that a weighted fraction $t$ of the branches will not satisfy the precondition $(Q_m, \lambda)$.
  Note that $t\in[0,1]$ since $\lambda\leq \lambda'$. %Now we proceed the soundness proof for the error parameters.
  For $m\notin S(\rho)$, only the trivial upper bound $\epsilon_m=1$ is guaranteed.
  Therefore, for each state $\rho$,
  \begin{align}
    \frac{1}{2}\nm{\sem{\widetilde{P}}\rho-\sem{P}\rho}_1
    &= \frac{1}{2} \nm{\sum_m \big(\sem{\widetilde{P_m}}(M_m\rho M_m^\dag) - \sem{P_m}(M_m\rho M_m^\dag)\big)}_1 \\
    &\leq \frac{1}{2} \sum_m\nm{\sem{\widetilde{P_m}}(M_m\rho M_m^\dag) - \sem{P_m}(M_m\rho M_m^\dag)}_1 \\
    &\leq \sum_{m\in S(\rho)}\tr(M_m\rho M_m^\dag)\epsilon+\sum_{m\notin S(\rho)}\tr(M_m\rho M_m^\dag) \\
    &\leq (1-t) \epsilon +t      = ((1-t) \epsilon +t )\tr(\rho).
  \end{align}
  Rewriting $\lambda=1-\delta$ for ease of notation, we have $\lambda'=1-t\delta$.
  Finally, we note that the case $t=0$ implies that the precondition of each branch is satisfied,
  and the error of the case statement can be bounded by $\epsilon$.

\item \emph{While-Bounded:} Let $P\equiv\qwhile{M[q]=1}{P_1}$. % and the error be $\epsilon$.
Let $S_k$ be the bounded while loop of $k$ iterations.
Define the linear maps $\E(\rho) := \sem{P_1}(M_1\rho M_1^\dag)$ and $\widetilde{\E}(\rho):=\sem{\widetilde{P_1}}(M_1\rho M_1^\dag)$ and
let $\sem{S_k}\rho=M_0\rho M_0^\dag+\sem{S_{k-1}}(\E(\rho))$ for $k\geq 1$ (with $\sem{S_0}\rho=\rho$).
In order to bound the distance between $\sem{P}$ and $\sem{\widetilde{P}}$, we first upper bound the distance between $\sem{S_k}$ and $\sem{\widetilde{S_k}}$ and then take the limit as $k\to\infty$.
We then have
\begin{align}
 \frac{1}{2} \nm{\sem{\widetilde{S_k}}\rho-\sem{S_k}\rho}_1
  &\leq \frac{1}{2}\nm{\sem{\widetilde{S}_{k-1}}(\widetilde{\E}(\rho))-\sem{S_{k-1}}(\E(\rho))}_1 \\\label{eq:while-1}
  &\leq \frac{1}{2} \nm{\widetilde{\E}(\rho)-\E(\rho)}_1+ \frac{1}{2} \nm{\sem{\widetilde{S}_{k-1}}(\E(\rho))-\sem{S_{k-1}}(\E(\rho))}_1 \\\label{eq:while-2}
  &\leq \frac{1}{2}\sum_{i=0}^{k-1}\nm{\widetilde{\E}(\E^i(\rho))-\E^{i+1}(\rho)} \\\label{eq:while-3}
  &\leq \epsilon\sum_{i=0}^{k-1}\tr(M_1^\dag M_1 \E^i(\rho)) \\\label{eq:while-4}
  &\leq n\epsilon\tr(M_1^\dag M_1\rho)\frac{1-a^{\lceil k/n\rceil}}{1-a}.
\end{align}
The second inequality \eq{while-1} follows from a technique similar to the one used in the proof of the \emph{Sequence} rule. 
We bound the first term in \eq{while-1} by $\epsilon\tr(M_1\rho M_1^\dag)$ by applying the premise $\rjudgment{Q}{\lambda}{P_1}{\epsilon}$ and the induction hypothesis.
We will prove the post-measurement state $M_1\rho M_1^\dagger$ indeed satisfies the precondition $(Q, \lambda)$ later.  
Similarly, each term in \eq{while-2} is bounded above by $\epsilon\tr(M_1^\dag M_1 \E^i(\rho))$ and thus the inequality in \eq{while-3} holds.

To establish the inequality in \eq{while-4}, let $b_i:=\tr(M_1^\dag M_1 \E^i(\rho))$. Then the sequence $\{b_k\}_k$ is non-negative and non-increasing.
We now prove an upper bound of the series.
Since $P$ is $(a,n)$-bounded, we know that
\begin{align}
  \tr(M_1^\dag M_1\E^n(\sigma)) = \tr((\E^*)^n(M_1^\dag M_1)\sigma)\leq  a \tr(M_1^\dag M_1\sigma), \text{ where } a<1
\end{align}
for every state $\sigma$, and therefore $b_{i+n}\leq a b_i$ for every $i$.
Since $b_k$ is non-increasing, we know that
\begin{align}
  \sum_{i=0}^{k-1} b_i
  &=\sum_{m=0}^{\lceil k/n \rceil-1} (b_{nm}+\ldots+b_{nm+n-1}) \\
  &\leq n\sum_{m=0}^{\lceil k/n\rceil-1}b_{nm}
  \leq n\sum_{m=0}^{\lceil k/n\rceil-1}a^m b_0
  =\frac{n(1-a^{\lceil k/n\rceil})b_0}{1-a}.
\end{align}
Thus the inequality in \eq{while-4} holds.
Since $b_0=\tr(M_1^\dag M_1\rho)\leq\tr(\rho)$,  we have
\begin{align}
 \frac{1}{2} \nm{\sem{\widetilde{S_k}}\rho-\sem{S_k}\rho}_1\leq \frac{n\epsilon(1-a^{\lceil k/n\rceil})}{1-a}\tr(\rho).
\end{align}
Taking the limit as $k\to\infty$, we have that $a^{\lceil k/n\rceil}\to 0$, which shows that $\frac{1}{2} \nm{\sem{\widetilde{S_k}}\rho-\sem{S_k}\rho}_1\leq \frac{n\epsilon}{1-a}\tr(\rho)$, as desired.

In order to apply the premise $\rjudgment{Q}{\lambda}{P_1}{\epsilon}$ to states of the form $M_1\E^i(\rho)M_1^\dag$, we need to show that $\tr(QM_1\E^i(\rho)M_1^\dag)\geq\lambda\tr(M_1\E^i(\rho)M_1^\dag)$ for each $i$.
We can prove this by induction.
For the base case $i=0$, by the precondition $(\lambda M_0^\dag M_0+M_1^\dag Q M_1, \lambda)$ on the input state to the loop, we have $\tr(M_1^\dag Q M_1\rho)\geq\lambda\tr(\rho)-\lambda\tr(M_0^\dag M_0\rho)=\lambda\tr(M_1^\dag M_1\rho)$.
Therefore, $\tr(QM_1\E^0(\rho)M_1^\dag)\geq\lambda\tr(M_1\E^0(\rho)M_1^\dag)$

For the inductive step, observe that the Hoare triple $\{Q\}P_1\{R\}$ yields
$\tr(R\sem{P_1}\sigma)\geq\tr(Q\sigma)$ for $R\equiv\lambda M_0^\dag M_0+M_1^\dag Q M_1$ and all states $\sigma$.
By the induction hypothesis, we have 
\[
\tr(R\sem{P_1}(M_1\E^i(\rho) M_1^\dag))\geq\tr(Q M_1 \E^i(\rho)M_1^\dag)\geq\lambda\tr( M_1 \E^i(\rho)M_1^\dag)\geq\lambda\tr(\E^{i+1}(\rho)),
\]
where the last inequality holds because quantum operations are trace-non-increasing (applied to $\sem{P_1}$). Note that $\sem{P_1}(M_1\E^i M_1^\dag) =\E^{i+1}$. Now by substituting $R$,
we have 
\begin{align}
  \tr(M_1^\dag Q M_1 \E^{i+1}(\rho))\geq \lambda\tr(\E^{i+1}(\rho))-\lambda\tr(M_0^\dag M_0 \E^{i+1}(\rho))=\lambda\tr(M_1^\dag M_1 \E^{i+1}(\rho)).
  \end{align}
Or equivalently, $ \tr( Q M_1 \E^{i+1}(\rho)M_1^\dag) \geq \lambda\tr( M_1 \E^{i+1}(\rho)M_1^\dag)$, which concludes the proof. 
\item \emph{While-Unbounded:} the proof is trivial as we use the trivial upper bound $1$. 
\end{enumerate}
\end{proof}

\section{Case Studies} \label{sec:case-studies}
In this section, we apply our robustness definition and logic to two example quantum programs: the quantum Bernoulli factory and the quantum walk.
We also demonstrate the (in)effectiveness of different error correction schemes on single-qubit errors and analyze the robustness of a fault-tolerant version of QBF.
Some computational details are deferred to the appendices. 

\subsection{Quantum Bernoulli Factory}\label{sec:qbf}
The quantum Bernoulli factory (QBF) \citep{DJR2015} is the quantum equivalent of the classical Bernoulli factory problem \citep{KO1994}.
In the classical Bernoulli factory problem, given a function $f:[0,1]\mapsto[0,1]$ and a coin that returns heads with unknown probability $p$, the goal is to simulate a new coin that returns head with probability $f(p)$.
In QBF, the goal is to generate the state $\ket{f(p)}$ given a description of $f$ and a \emph{quantum coin} described by the state
\begin{equation}\label{eq:inppp}
\ket{p}:=\sqrt{p}\ket{0}+\sqrt{1-p}\ket{1}.
\end{equation}

QBF is interesting because it can simulate a strictly larger class of functions $f$ than can be simulated by its classical counterpart.
One example of a function that can be simulated by QBF, but not by the classical Bernoulli factory, is the probability amplification function.
The key to simulating the probability amplification function is producing the state $\ket{f(p)}=(2p-1)\ket{0}+2\sqrt{p(1-p)}\ket{1}$.
This state can be prepared by (the ideal version of) the following program \citep{LY2017}:
\begin{align*}
  \widetilde{QBF}\equiv~ & q_1:=\ket{1};~~q_2:=\ket{1}; \\
                      & \qwhile{M[q_2]=1}{ \\
                      & \qquad q_1:=\ket{0};~~q_2:=\ket{0};\\
                      & \qquad q_1:\cong_{p_V,\Phi_V}V[q_1];~~q_2:\cong_{p_V,\Phi_V}V[q_2]; \\
                      & \qquad q_1,q_2:\cong_{p_U,\Phi_U}U[q_1,q_2]},
\end{align*}
where $M$ is the standard basis measurement $\{ \ket{0} \bra{0}, \ket{1}\bra{1} \}$.
The unitary $U$ is defined by
\begin{align}
  U = \ket{01}\bra{\phi^+}+\ket{00}\bra{\phi^-}+\ket{10}\bra{\psi^+}+\ket{11}\bra{\psi^-}
  % U\ket{\Phi^+}=\ket{01}; && U\ket{\Phi^-}=\ket{00}; && U\ket{\Psi^+}=\ket{10}; && U\ket{\Psi^-}=\ket{11}.
\end{align}
where $\ket{\phi^\pm} = \frac{1}{\sqrt{2}}(|00\rangle\pm|11\rangle)$ and $\ket{\psi^\pm} = \frac{1}{\sqrt{2}}(|01\rangle\pm|10\rangle)$.
The unitary $V:=\left[\begin{array}{cc}\sqrt{p} & -\sqrt{1-p} \\ \sqrt{1-p} & \sqrt{p}\end{array}\right]$ acting on $\ket{0}$ generates the state $\ket{p}$.
We denote the error due to noisy unitary application of $V$ and $U$ by
$\epsilon_V=p_V\nm{\Phi_V-V\circ V^\dag}_\diamond$ and $\epsilon_U=p_U\nm{\Phi_U-U\circ U^\dag}_\diamond$ respectively.
Note that to simplify our discussion we compute $\epsilon_V$ and $\epsilon_U$ using the trivial precondition $(I,0)$. 
Recall that $\nm{\cdot}_{I,0} = \nm{\cdot}_\diamond$.

Now we prove that $\rjudgment{I}{0}{\widetilde{QBF}}{4\epsilon_V+2\epsilon_U}$.
First, by the \emph{Sequence} and \emph{Unitary} rules, we bound the error in the loop body by $2\epsilon_V+\epsilon_U$.
To show $(a,n)$-boundedness of the loop, we must consider the behavior of $\E^*$ where 
\begin{align}
\E^*= (M_1\circ M_1^\dag)^*\circ\sem{q_1:=\ket{p};q_2:=\ket{p}}^*\circ  \sem{U[q_1,q_2]}^*.
\end{align}
We evaluate $\E^*(M_1^\dag M_1)$ as follows.
Recall that $M_1=I\otimes \ket{1}\bra{1}$.
\begin{align}
  M_1^\dag M_1=I\otimes\ket{1}\bra{1}\nonumber
  & \xmapsto{\sem{U[q_1,q_2]}^*} \ket{\phi^+}\bra{\phi^+}+\ket{\psi^-}\bra{\psi^-} \\
  & \xmapsto{\sem{q_1:=\ket{p};q_2:=\ket{p}}^*} \frac{1}{2}I\otimes I \\
  &\xmapsto{(M_1\circ M_1^\dag)^*} \frac{1}{2}M_1^\dag M_1.\label{eq:qbf-boundedness}
\end{align}
This shows that $\E^*(M_1^\dag  M_1)=\frac{1}{2}M_1^\dag M_1$.
Therefore, the while loop is $(\frac{1}{2},1)$-bounded.
Now we can use the \emph{While-Bounded} rule to conclude that the error bound for the loop is $\frac{2\epsilon_V+\epsilon_U}{1-\frac{1}{2}} = 4\epsilon_V+2\epsilon_U$,
and therefore $\rjudgment{I}{0}{(\qwhile{M[q_2]=1}{\widetilde{P_1}})}{4\epsilon_V+2\epsilon_U}$ where $\widetilde{P_1}$ is the body of the loop.
Two additional applications of the \emph{Sequence} rule conclude the proof.

We can compute $\epsilon_V$ and $\epsilon_U$ by identifying the error probabilities $p_V, p_U$ and error models $\Phi_V, \Phi_U$, and numerically calculating the diamond norm.
For example, consider the case where the state preparation is ideal, i.e., $p_V=0$,
and noise in the application of $U$ is characterized by $p_U=10^{-5}$ and $\Phi_U = ( \frac{X\circ X + Y\circ Y + Z\circ Z + I\circ I}{4})^{\otimes 2}$, %%\frac{I}{4}  $
the $4$-dimensional depolarizing channel. %with the depolarizing probability $p'=\frac{1}{2}$. %$\Phi_U=I\circ I$.
Then $\epsilon_V=0$ and $\epsilon_U=p_U\nm{\Phi_U-U\circ U^\dag}_\diamond$.
Applying an SDP solver \citep{dS15,W09} for the calculation of the diamond norm, we find that $\epsilon_U=1 \times 10^{-5} \times 0.9375 = 9.375 \times 10^{-6}$,
and therefore the error bound of $\widetilde{QBF}$ is $1.875\times 10^{-5}$.
More details can be found in \app{tedious1}.

\subsection{Quantum Walk on a Circle}

Here we analyze the error of the quantum walk algorithm introduced in \sec{quantum-while-language}.
The noisy quantum walk on a circle with $n$ points can be written as the following program:
\begin{align}
  \widetilde{QW}_n\equiv p:=\ket{0}; c:=\ket{L}; \qwhile{M[p]= 1}{c:\cong_{p_H,\Phi_H}H[c]; c,p:\cong_{p_S,\Phi_S}S[c,p]}.
\end{align}

Now we show that $\rjudgment{I}{0}{\widetilde{QW}_6}{30(\epsilon_H+\epsilon_S)}$ where $\epsilon_H=p_H\nm{\Phi_H-H\circ H^\dag}_\diamond$ and $\epsilon_S=p_S\nm{\Phi_S-S\circ S^\dag}_\diamond$ are the errors due to noisy application of $H$ and $S$ respectively.
First, by the \emph{Sequence} and \emph{Unitary} rules, we bound the error in the loop body by $\epsilon_H+\epsilon_S$.
Next, we numerically test increasing values of $(a,n)$ until we find a pair that satisfies \eq{anbounded}.
For this program, we find that the pair $(\frac{5}{6},5)$ satisfies the inequality, i.e., $(\E^*)^5(M_1^\dag M_1)\sqsubseteq \frac{5}{6}M_1^\dag M_1$ where $\E^*=(M_1^\dag\circ M_1)\circ\sem{c,p:=S[c,p]}^*\circ \sem{c:=H[c]}^*$.
This implies that the loop is $(\frac{5}{6},5)$-bounded.
Note that here, unlike in the previous example, we have computed the $(a,n)$ values numerically.
This may be useful in cases where direct deduction of $a$ and $n$ is difficult.
Now we can use the \emph{While-Bounded} rule to conclude that the error bound for the loop is $\frac{5(\epsilon_H+\epsilon_S)}{1-\frac{5}{6}} = 30(\epsilon_H+\epsilon_S)$,
and therefore $\rjudgment{I}{0}{(\qwhile{M[p]=1}{\widetilde{P_1}})}{30(\epsilon_H+\epsilon_S)}$ where $\widetilde{P_1}$ is the body of the loop.
Two additional applications of the \emph{Sequence} rule conclude the proof.

As an example, consider the program $\widetilde{QW}_6$ where only the Hadamard gate may be faulty, i.e., $p_S=0$. 
Say that noisy Hadamard application is characterized by $p_H=5\times 10^{-5}$ and $\Phi_H (\rho)=  \frac{X\circ X + Y\circ Y + Z\circ Z + I\circ I}{4} $, the $2$-dimensional depolarizing channel. 
Applying an SDP solver \citep{dS15,W09,W13} for the calculation of the diamond norm, we find that $\epsilon_H=5\times 10^{-5} \times 0.75 = 3.75 \times 10^{-5}$.
Therefore the error of $\widetilde{QW}_6$ is $30\epsilon_H = 1.125 \times 10^{-3}$.
More details can be found in \app{tedious1}.

\subsection{Error Correction}\label{sec:ex-qecc}

In this section we use our semantics to show that an error correction scheme that is appropriate for the error model can reduce noise in a program, while an inappropriate error correction scheme may do the opposite.  
We consider three programs: $P_1$, $P_2$, and $P_3$.
$P_1$ performs the identity operation, i.e., $P_1 \equiv q := I[q]$.
The noisy version of $P_1$, $\widetilde{P_1}$, allows error during the application of $I$.
$P_2$ and $P_3$ are semantically equivalent to $P_1$ (in the sense that they correspond to the same superoperator\footnote{To make this mathematically rigorous, one needs to distinguish between variables and ancillas in the program. The semantics refers to the superoperator that traces out the ancilla part.}), but both employ error correction.
$P_2$ uses a three-qubit repetition code that can correct a bit flip ($X$ error) on a single qubit and
$P_3$ uses a three-qubit repetition code that can correct a phase flip ($Z$ error) on a single qubit.
$P_2$ and $P_3$ both have the following form: 
\begin{align*}
  &\overline{q} := ENCODE[q]; \\
  &\overline{q} := \overline{I}[\overline{q}]; \\
  &\overline{q} := CORRECT[\overline{q}]; \\
  &q := DECODE[\overline{q}]
\end{align*}
where we use $ENCODE$ as a stand-in for the operations associated with turning the qubit $q$ into its encoded counterpart $\overline{q}$, $CORRECT$ as a stand-in for the operations associated with syndrome detection and error correction, and $DECODE$ as a stand-in for the operations associated with converting the encoded qubit $\overline{q}$ back into $q$. 
We use $\overline{I}$ to represent a fault-tolerant version of the identity operation.
In the noisy programs $\widetilde{P_2}$ and $\widetilde{P_3}$, we only allow error during application of $\overline{I}$.
% \shh{Make statement on the assumption.}
To simplify the calculation, we will assume that encoding, decoding, and error correction are all ideal (i.e. they have an error probability of 0).
For a description of $ENCODE$, $CORRECT$, and $DECODE$ see \app{rep-codes}.

For all three programs, we define the error model acting on a single qubit by $\Phi(\rho) = X\rho X$.
With this error model, application of an $I$ gate to a qubit will succeed with probability $1-p$, and will instead become an application of a $X$ gate with probability $p$.
Now we can use our definition of the denotational semantics from \sec{noisy-progs} to directly compute $\sem{\widetilde{P_1}}$, $\sem{\widetilde{P_2}}$, and $\sem{\widetilde{P_3}}$.
We find that:
\begin{align}
  \sem{\widetilde{P_1}}\rho &= (1-p) I\rho I + p X\rho X \\
  \sem{\widetilde{P_2}}\rho &= ((1-p)^3 + 3p(1-p)^2) I\rho I + (3p^2(1-p) + p^3) X\rho X \\
  \sem{\widetilde{P_3}}\rho &= ((1-p)^3 + 3p^2(1-p)^2) I\rho I + (3p(1-p) + p^3) Z\rho Z
\end{align}
Now we can directly compute $\nm{\sem{P_1} - \sem{\widetilde{P_1}}}_{\diamond}$, $\nm{\sem{P_2} - \sem{\widetilde{P_2}}}_{\diamond}$, and $\nm{\sem{P_3} - \sem{\widetilde{P_3}}}_{\diamond}$ to determine error rates for these programs. 
We find that:
\begin{align}
  \nm{\sem{P_1} - \sem{\widetilde{P_1}}}_{\diamond} &= p \\
  \nm{\sem{P_2} - \sem{\widetilde{P_2}}}_{\diamond} &= 3p^2 - 2p^3 \\
  \nm{\sem{P_3} - \sem{\widetilde{P_3}}}_{\diamond} &= 3p(1-p)^2 + p^3
\end{align}
This allows us to conclude that, under $Q=I$ and $\lambda=0$,
\[\widetilde{P_1}~\text{is $p$ robust} \qquad \widetilde{P_2}~\text{is $3p^2 - 2p^3$ robust} \qquad \widetilde{P_3}~\text{is $3p(1-p)^2 + p^3$ robust}\]
Note that for $0 < p < \frac{1}{2}$, we have that $3p^2 - 2p^3 < p$ and $p < 3p(1-p)^2 + p^3$, which tells us that, in the presence of $X$ errors, correcting for bit flips will improve the error rate while correcting for phase flips will make the error rate worse.
For computational details see \app{ec-example}.

Here we invoke the semantics of the program to prove quantum robustness by definition. 
This is an expensive calculation.
However, it is necessary to account for the effect of error correction. 
Ideally, we imagine a combination of uses of the semantics and the rules to trade off between the cost of the calculation and the accuracy of the bounds.

\subsection{Fault-tolerant Quantum Bernoulli Factory}\label{sec:ft-qbf}

In this section, we consider a fault-tolerant implementation of the quantum Bernoulli factory.
The fault-tolerant $QBF$, denoted by $\overline{QBF}$, can be written as follows.
\begin{align*}
  \widetilde{\overline{QBF}}\equiv & q_1:=\ket{1};~~q_2:=\ket{1}; \\
                                    & \qwhile{M[q_2]=1}{ \\
                        & \qquad q_1:=\ket{0};~~q_2:=\ket{0}; \\
                        & \qquad \overline{q}_1:=ENCODE[q_1];~~\overline{q}_2:=ENCODE[q_2]; \\
                        & \qquad \overline{q}_1:\cong_{1,\Phi_{\overline{V}}}\overline{V}[\overline{q}_1]; \\
                        & \qquad \overline{q}_1:=CORRECT[\overline{q}_1]; \\
                        & \qquad \overline{q}_2:\cong_{1,\Phi_{\overline{V}}}\overline{V}[\overline{q}_2]; \\
                        & \qquad \overline{q}_2:=CORRECT[\overline{q}_2]; \\
                        & \qquad \overline{q}_1;\overline{q}_2 :\cong_{1,\Phi_{\overline{U}}}\overline{U}[\overline{q}_1,\overline{q}_2];\\
                        & \qquad \overline{q}_1:=CORRECT[\overline{q}_1];~~
                          \overline{q}_2:=CORRECT[\overline{q}_2]; \\
                        & \qquad q_1:=DECODE[\overline{q}_1];~~q_2:=DECODE[\overline{q}_2]\\
                        & \qquad  }
\end{align*}
where $ENCODE$, $CORRECT$, and $DECODE$ have the same meanings as in the previous section.
As before, we assume that encoding, decoding, and error correction are not affected by noise. %\xw{comment more}
$\overline{V}$ and $\overline{U}$ are the fault-tolerant operators that correspond to $V$ and $U$ respectively.
We define the error model associated with $\overline{V}$ by $\Phi_{\overline{V}}(\rho)=\Phi_V^{\otimes 3}(\overline{V}\rho\overline{V}^\dag)$ where $\Phi_V(\rho) = (1-p_V)I\rho I + p_V X\rho X$.
Similarly, we define the error model associated with $\overline{U}$ by $\Phi_{\overline{U}}(\rho)=\Phi_U^{\otimes 6}(\overline{U}\rho\overline{U}^\dag)$ where $\Phi_U(\rho) = (1-p_U)I\rho I + p_U X\rho X$.
Note that because we are using a three-qubit repetition code, $\overline{V}$ and $\overline{U}$ will be 3-qubit and 6-qubit unitaries respectively. 
These noise models are applied with probability one.

Using the definitions from \sec{noisy-progs} we can show that for $i\in\{1,2\}$,
\begin{align}\label{eq:v-sem}
  \sem{
        \overline{q}_i:\cong_{1,\Phi_{\overline{V}}}\overline{V}[\overline{q}_i];~~
        \overline{q}_i:=CORRECT[\overline{q}_i]}\rho
          =(1-q_V)\overline{V}\rho \overline{V}^\dag+q_VX^{\otimes 3}\overline{V}\rho \overline{V}^\dag X^{\otimes 3},
\end{align}
where $q_V=3p_V^2(1-p_V)+p_V^3$.
Note that \eq{v-sem} is obtained through a calculation similar to the one used to compute $\sem{P_2}$ in \sec{ex-qecc}. 
This says applying $\overline{V}$ followed by error correction is equivalent to applying $\overline{V}$ with probability $1 - q_V$, and applying $\overline{V}$ followed by $X^{\otimes 3}$ with probability $q_V$.
Thus we can use the \emph{Unitary} rule to compute the error of the noisy application of $\overline{V}$ followed by error correction:
\begin{align}
  \epsilon_{\overline{V}}
  &=q_V\nm{(X^{\otimes 3}\circ X^{\otimes 3})\circ(\overline{V}\circ\overline{V}^\dag)-(\overline{V}\circ\overline{V}^\dag)}_\diamond
  =q_V\nm{(X^{\otimes 3}\circ X^{\otimes 3})-(I\circ I)}_\diamond = q_V.
\end{align}
The final equality holds because the unitaries are perfectly distinguishable.
Note that $\nm{\Phi(U\circ U^\dag)-U\circ U^\dag}_{\diamond}=\nm{\Phi-I\circ I}_\diamond$ holds for any superoperator $\Phi$ and unitary $U$.
Similarly, 
\begin{align}
  \sem{
    & \overline{q}_1,\overline{q}_2:\cong_{1,\Phi_{\overline{U}}}\overline{U}[\overline{q}_1,\overline{q}_2];~~
    \overline{q}_1:=CORRECT[\overline{q}_1];~~
    \overline{q}_2:=CORRECT[\overline{q}_2]
      }\rho \\
    &\qquad =((1-q_U)I\circ I+q_U X^{\otimes 3}\circ X^{\otimes 3})^{\otimes 2}(\overline{U}\rho \overline{U}^\dag),
\end{align}
where $q_U=3p_U^2(1-p_U)+p_U^3$.
Thus, by the \emph{Unitary} rule, the error of the noisy application of $\overline{U}$ followed by error correction is
\begin{align}
  \epsilon_{\overline{U}} &= \nm{((1-q_U)I\circ I+q_U X^{\otimes 3}\circ X^{\otimes 3})^{\otimes 2}-I\circ I}_\diamond \\\label{eq:ft-error-u}
             &= \nm{q_U^2(X^{\otimes 6}\circ X^{\otimes 6})+q_U(1-q_U)(I^{\otimes 3}\otimes X^{\otimes 3})\circ(I^{\otimes 3}\otimes X^{\otimes 3}) \\
             &\qquad +q_U(1-q_U)(X^{\otimes 3}\otimes I^{\otimes 3})\circ(X^{\otimes 3}\otimes I^{\otimes 3})-(2q_U-q_U^2)I\circ I }_\diamond \\
             & = 2q_U-q_U^2.
\end{align}
The last equality above can be proved as follows.
For ease of notation, we define the superoperator
\begin{align}
  \E(\rho):=& \frac{1}{2q_U-q_U^2}(q_U^2(X^{\otimes 6}\rho X^{\otimes 6})+q_U(1-q_U)(I^{\otimes 3}\otimes X^{\otimes 3})\rho(I^{\otimes 3}\otimes X^{\otimes 3}) \\
  & \qquad\qquad +q_U(1-q_U)(X^{\otimes 3}\otimes I^{\otimes 3})\rho(X^{\otimes 3}\otimes I^{\otimes 3})).
\end{align}
The superoperator $\E$ is trace-preserving, so $\epsilon_{\overline{U}}=(2q_U-q_U^2)\nm{\E-I\circ I}_\diamond\leq (2q_U-q_U^2)$.
To show $\epsilon_{\overline{U}}$ is also lower bounded by $2q_U-q_U^2$, it suffices to consider the input state $\rho=(\ket{0}\bra{0})^{\otimes 6}$ and the projector $\Pi=I-(\ket{0}\bra{0})^{\otimes 6}$.
By definition, $\epsilon_{\overline{U}}=(2q_U-q_U^2)\nm{\E-I\circ I}_\diamond\geq (2q_U-q_U^2)\tr(\Pi(\E(\rho)-\rho))=2q_U-q_U^2$.

Now using the argument given in \sec{qbf}, we can show that $\rjudgment{I}{0}{\widetilde{\overline{QBF}}}{4\epsilon_{\overline{V}}+2\epsilon_{\overline{U}}}$.
Note that $ENCODE$ and $DECODE$ do not impact the robustness.
Without error correction, as shown in \sec{qbf}, $\rjudgment{I}{0}{\widetilde{QBF}}{4\epsilon_V+2\epsilon_U}$ where $\epsilon_V=\nm{\Phi_V-I\circ I}_\diamond=p_V$ and
$\epsilon_U=\nm{\Phi_U^{\otimes 2}-I\circ I}_\diamond=2p_U-p_U^2$.
In the case where the error rate is constant, i.e., $p_V=p_U=p$ for some probability $p$,
$\rjudgment{I}{0}{\widetilde{QBF}}{O(p)}$ and $\rjudgment{I}{0}{\widetilde{\overline{QBF}}}{O(p^2)}$.
This shows that the error of $\widetilde{QBF}$ is suppressed by a factor of $p$ with a fault-tolerant implementation of the loop body.

\begin{remark}
Throughout our example, we make the assumption that operations like $ENCODE$, $CORRECT$ and $DECODE$ are noise-free. 
In the actual setting of fault-tolerant quantum computation, they can also contain noise.   
In order to suppress the error rate, one may need to use the construction of fault-tolerant gadgets in the proof of the threshold theorem \citep{AB97}. 
We remark that the assumption we made is just to simplify the example and ease presentation. 
It would in fact be possible to prove a threshold theorem in our formalism without this assumption, although the calculation could be much more complicated. 
It is an interesting open question to see whether one can simplify this calculation by adding more rules to our logic. 
\end{remark}

\section{Conclusions and Future Work}

We have presented a semantics for describing quantum computation with errors and an analysis that bounds the distance between the result of a noisy program and its corresponding ideal program on the same input.
We used our analysis to compute error bounds for noisy versions of the quantum Bernoulli factory and quantum walk programs.
We also showed how our analysis can be used to compute the error bounds for small circuits with and without error correction, showing examples of when using error correction is beneficial and when there are tradeoffs between the efficiency of error corrections and related costs. 

A natural next step for our work is to encode the rules from \sec{noisy-rules} in a proof assistant so that they can be applied in an automated fashion.
We have shown that the $(Q,\lambda)$-diamond norm can be computed by an SDP, and the $(a,n)$ values for loop boundedness can be computed analytically or numerically, so implementing our rules is feasible.
Given an implementation, we are interested how our analysis may be used to construct circuits with lower error rates.
This application is inspired by work by \citet{Misailovic14}, which uses classical reliability analysis to determine which operations can be replaced by their noisy counterparts.
It could hence be used to inform decisions about which implementations of quantum algorithms are practical for near-term use.

%% Acknowledgments
\begin{acks}                            %% acks environment is optional
We would like to thank Andrew Childs for helpful discussions on quantum walks. 
This material is based upon work supported by the 
\grantsponsor{ascr}{U.S. Department of Energy, Office of Science, 
Office of Advanced Scientific Computing Research}{http://dx.doi.org/10.13039/100006192}  
Quantum Testbed Pathfinder Program under Award Number \grantnum{ascr}{DE-SC0019040},
and the \grantsponsor{ciar}{Canadian Institute for Advanced Research}{http://dx.doi.org/10.13039/100007631}.

%   %% Commands \grantsponsor{<sponsorID>}{<name>}{<url>} and
%   %% \grantnum[<url>]{<sponsorID>}{<number>} should be used to
%   %% acknowledge financial support and will be used by metadata
%   %% extraction tools.
%   This material is based upon work supported by the
%   \grantsponsor{GS100000001}{National Science
%     Foundation}{http://dx.doi.org/10.13039/100006192} under Grant
%   No.~\grantnum{GS100000001}{nnnnnnn} and Grant
%   No.~\grantnum{GS100000001}{mmmmmmm}.  Any opinions, findings, and
%   conclusions or recommendations expressed in this material are those
%   of the author and do not necessarily reflect the views of the
%   National Science Foundation.

\end{acks}

%% Bibliography
\bibliography{references}

% \iftechrep
%% Appendix
\appendix

\newpage
 \section{Appendix}\label{sec::appendix}

\subsection{Proof that the \texorpdfstring{$(Q,\lambda)$-diamond}{(Q,lambda)-diamond} Norm is a Seminorm} \label{app:Qlambda}
Here we prove that $(Q, \lambda)$-diamond norm is a seminorm. Recall its definition. 

\begin{definition}[$(Q,\lambda)$-diamond norm] 
Given superoperators $\E$, $\E'$, quantum predicate $Q$ over $\H$, and $0\leq \lambda \leq 1$, the $(Q, \lambda)$-diamond norm between $\E$ and $\E'$, denoted $\nm{\E-\E'}_{Q,\lambda}$,  is defined by
\begin{align}
  \nm{\E-\E'}_{Q,\lambda} \equiv \max_{\rho\in\D(\H\otimes\mathcal{A})~:~\tr(\rho) = 1,~\tr(Q\rho)\geq\lambda} \tdist{\E\otimes I_{\mathcal{A}}(\rho)}{\E'\otimes I_{\mathcal{A}}(\rho)},
\end{align}
where $\A$ is any auxiliary space and can be assumed to be $\H$ without loss of generality.
\end{definition}

\begin{proof}
Recall from \eq{tracedist} that $0\leq \tdist{\E\otimes I_{\mathcal{A}}(\rho)}{\E'\otimes I_{\mathcal{A}}(\rho)}\leq 1$ for any quantum state $\rho$, so $ \nm{\E-\E'}_{Q,\lambda}\geq 0$. 

\begin{enumerate}
\item \emph{Positive scalability:}
This property is inherited directly from the diamond norm. One may also observe that for any $\alpha\in\mathbb{C} $,
\begin{align*}
\tdist{\E\otimes I_{\mathcal{A}}(\alpha\rho)}{\E'\otimes I_{\mathcal{A}}(\alpha\rho)} &= \tdist{\alpha(\E\otimes I_{\mathcal{A}}(\rho))}{\alpha(\E'\otimes I_{\mathcal{A}}(\rho))}\\& = \alpha\tdist{\E\otimes I_{\mathcal{A}}(\rho)}{\E'\otimes I_{\mathcal{A}}(\rho)}.
\end{align*}

\item \emph{Triangle Inequality:}
For quantum superoperators  $\E,\E', \F,\F'$, we have 
\begin{align*}
\nm{(\E-\E')+(\F-\F')}_{Q,\lambda} &= \max_{ \substack{\rho\in\D(\H\otimes\mathcal{A})~:\\\tr(\rho) = 1,~\tr(Q\rho)\geq\lambda} } \tdist{(\E+\F)\otimes I_{\mathcal{A}}(\rho)}{(\E'+\F')\otimes I_{\mathcal{A}}(\rho)}.\\ 
\end{align*}

Suppose the maximal value on the right hand side is attained at a state $\rho^*\in\D(\H\otimes\mathcal{A})$. Note that, by definition, this requires $\rho^*$ to satisfy $\tr(\rho^*) = 1, \tr(Q\rho^*)\geq\lambda$.  Then,

\begin{align*}
\nm{(\E-\E')+(\F-\F')}_{Q,\lambda} &= \tdist{(\E+\F)\otimes I_{\mathcal{A}}(\rho^*)}{(\E'+\F')\otimes I_{\mathcal{A}}(\rho^*)}\\
&= \max_{0\sqsubseteq P\sqsubseteq I} \tr(P((\E+\F)\otimes I_{\mathcal{A}}(\rho^*) - (\E'+\F')\otimes I_{\mathcal{A}}(\rho^*)))\\
&=\max_{0\sqsubseteq P\sqsubseteq I} \tr(P((\E-\E')\otimes I_{\mathcal{A}}(\rho^*) + (\F-\F')\otimes I_{\mathcal{A}}(\rho^*)))\\
&\leq \max_{0\sqsubseteq P_1\sqsubseteq I} \tr(P_1((\E-\E')\otimes  I_{\mathcal{A}}(\rho^*))) + \max_{0\sqsubseteq P_2\sqsubseteq I}\tr(P_2(\F-\F')\otimes I_{\mathcal{A}}(\rho^*)))\\
&\leq \max_{\substack{\rho_1\in\D(\H\otimes\mathcal{A})~:\\\tr(\rho_1) = 1,~\tr(Q\rho_1)\geq\lambda}}\Big(\max_{0\sqsubseteq P_1\sqsubseteq I} \tr(P_1((\E-\E')\otimes I_{\mathcal{A}}(\rho_1)))\Big)\\
&\qquad + \max_{\substack{\rho_2\in\D(\H\otimes\mathcal{A})~:\\\tr(\rho_2) = 1,~\tr(Q\rho_2)\geq\lambda}}\Big(\max_{0\sqsubseteq P_2\sqsubseteq I} \tr(P_2((\F-\F')\otimes I_{\mathcal{A}}(\rho_2)))\Big)\\
&=\max_{\substack{\rho_1\in\D(\H\otimes\mathcal{A})~:\\\tr(\rho_1) = 1,~\tr(Q\rho_1)\geq\lambda}}\tdist{\E\otimes I_{\mathcal{A}}(\rho_1)}{\E'\otimes I_{\mathcal{A}}(\rho_1)} \\
&\qquad + \max_{\substack{\rho_2\in\D(\H\otimes I_\mathcal{A})~:\\\tr(\rho_2) = 1,~\tr(Q\rho_2)\geq\lambda}}\tdist{\F\otimes I_{\mathcal{A}}(\rho_2)}{\F'\otimes I_{\mathcal{A}}(\rho_2)}\\
&= \nm{\E-\E'}_{Q,\lambda} +  \nm{\F-\F'}_{Q,\lambda}. 
\end{align*}

\end{enumerate}
\end{proof}

\subsection{Computation Details for the \texorpdfstring{$\widetilde{QBF}$}{QBF} and \texorpdfstring{$\widetilde{QW_6}$}{QW} Case Studies}\label{app:tedious1}
First note that an $n$-qubit depolarizing channel can be written in Kraus form as
$\frac{1}{2^{2n}}(I\circ I+X\circ X+Y\circ Y+Z\circ Z)^{\otimes n}$.
To numerically compute the error $\epsilon_U$ for $\widetilde{QBF}$,
one first computes the Choi-Jamiolkowski representation of $\Phi_U-U\circ U^\dag$:  
\begin{align}
  J(\Phi_U-U\circ U^\dag) =
  & \left(\frac{1}{4}\sum_{i,k\in\{0,1\}} (X\ket{i}\bra{k} X + Y\ket{i}\bra{k} Y + Z\ket{i}\bra{k} Z + \ket{i}\bra{k})\otimes \ket{i}\bra{k}\right)^{\otimes 2}\\
  & - \sum_{i,j,k,l\in\{0,1\}}U\ket{ij}\bra{kl} U^\dag\otimes \ket{ij}\bra{kl}.\label{eq:choochoo}
\end{align}
Using the matrix found in $\eq{choochoo}$, 
we numerically compute the diamond norm \citep{dS15,W09} and obtain $\epsilon_U =p_U\nm{\Phi_U-U\circ U^\dag}_\diamond =1 \times 10^{-5} \times 0.9375$.
Thus the error of the entire loop is $2\epsilon_U = 1.875\times 10^{-5}$. 

For $\widetilde{QW_6}$, we compute the Choi-Jamiolkowski representation of $\Phi_H-H\circ H^\dag$: 

\begin{align}
  J(\Phi_H-H\circ H^\dag)
  =& \sum_{i,j\in \{0,1\}}\left(\frac{1}{4}\left(X\ket{i}\bra{j}X + Y\ket{i}\bra{j}Y + Z \ket{i}\bra{j}Z + \ket{i}\bra{j}\right)-H(\ket{i}\bra{j})H\right)\otimes \ket{i}\bra{j}.\label{eq:choochoo22}
\end{align}
Using the matrix found in \eq{choochoo22}, we compute the diamond norm \citep{dS15,W09} 
and obtain $\epsilon_H =p_H\nm{\Phi_H-H\circ H^\dag}_\diamond =5 \times 10^{-5} \times 0.75= 3.75 \times 10^{-5}$.
Thus the error of the entire loop is $30\epsilon_H =1.125\times 10^{-3}$. 

\subsection{Description of Quantum Repetition Codes} \label{app:rep-codes}

Quantum repetition codes are the quantum equivalent of classical repetition codes.
In a classical 3-bit repetition code, 0 is encoded as 000 and 1 is encoded as 111.
To do error correction, you can take a majority vote of the three bits (e.g. MAJ(0,0,1) = 0, so 001 will be corrected to 000).
Alternatively, you can use two \emph{parity checks} that indicate whether the first and second, and second and third, bits agree.
You can then use this information to do correction (e.g. if the first and second bits do not agree, but the second and third bits do, then the first bit must be in error).

\begin{figure}
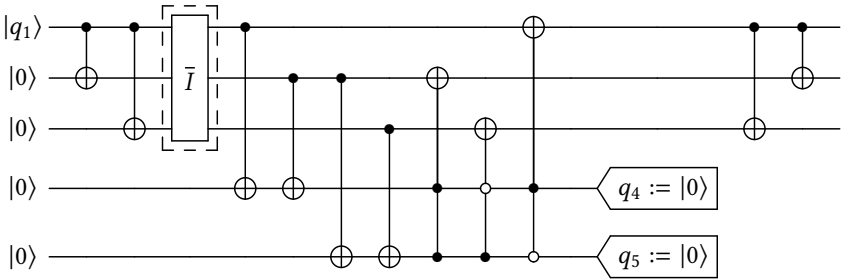

\begin{align*}
  & \qquad q_2:=\ket{0};~~q_3:=\ket{0};~~q_4:=\ket{0};~~q_5:=\ket{0}; & &\text{\bf Initialization}\\
  & \qquad q_1,q_2:= CNOT[q_1,q_2]; & &\text{\bf Encoding}\\
  & \qquad q_1,q_3:=CNOT[q_1,q_3];\\
  & \qquad q_1,q_2,q_3:={I}^{\otimes 3}[q_1,q_2,q_3]; & &\text{\bf Fault-tolerant gate application}\\
  & \qquad q_1,q_4:= CNOT[q_1,q_4]; & &\text{\bf Syndrome detection}\\
  & \qquad q_2,q_4:= CNOT[q_2,q_4];\\
  & \qquad q_2,q_5:= CNOT[q_2,q_5];\\
  & \qquad q_3,q_5:= CNOT[q_3,q_5];\\
  &\qquad q_4,q_5,q_2 := 11-TOFFOLI[q_4,q_5,q_2]; & &\text{\bf Error correction}\\
  &\qquad q_4,q_5,q_3 := 01-TOFFOLI[q_4,q_5,q_3];\\
  &\qquad q_4,q_5,q_1 := 10-TOFFOLI[q_4,q_5,q_1];\\
  &\qquad q_4 := \ket{0};~~q_5 := \ket{0}; & &\text{\bf Syndrome clean-up}\\
  & \qquad q_1,q_3:=CNOT[q_1,q_3]; & &\text{\bf Decoding}\\
  & \qquad q_1, q_2:=CNOT[q_1,q_2];\\
\end{align*}
\caption{Identity operation protected by a quantum repetition code that can correct one bit flip ($X$) error. }
\label{fig:ec-example}
\end{figure}

In the quantum case, to construct a code that corrects a single bit flip (i.e. $X$ error), we encode the state $\ket{\psi} = \alpha\ket{0} + \beta\ket{1}$ as $\ket{\overline{\psi}} = \alpha\ket{000} + \beta\ket{111}$.
We do error correction using the results of \emph{syndrome measurements}, which are similar to parity checks.
The program in \fig{ec-example} shows how to use the bit flip repetition code to protect an application of the identity gate.
Note that this is the program $\widetilde{P_2}$ from \sec{ex-qecc}.

On the right-hand side of the program in \fig{ec-example}, we have labeled every stage of the computation with its purpose.
The first line of the program (\emph{Initialization}) initializes qubits $q_2, q_3, q_4,$ and $q_5$ to the $\ket{0}$ state.
Assuming that qubit $q_1$ is initially in state $\alpha\ket{0} + \beta\ket{1}$, the state of system after the first line can be written as $(\alpha\ket{000} + \beta\ket{100})\ket{00}$.
The next two lines (\emph{Encoding}) apply $CNOT$ gates, which are the quantum equivalent of classical XOR gates, to construct the state $(\alpha\ket{000} + \beta\ket{111})\ket{00}$.
Note that at this point, qubits $q_1, q_2,$ and $q_3$ form a codeword.
The next line (\emph{Fault-tolerant gate application}) applies a fault-tolerant version of the identity gate, which in this case is simply $I^{\otimes 3}$.
The next four lines perform \emph{Syndrome detection}, and the three lines after that perform \emph{Error correction} based on the detected syndromes.
The next line (\emph{Syndrome clean-up}) resets the syndrome qubits to the $\ket{0}$ state,
and the final two lines (\emph{Decoding}) return the state $(\alpha\ket{000} + \beta\ket{111})\ket{00}$ to $(\alpha\ket{000} + \beta\ket{100})\ket{00}$.
The $ENCODE$ notation from \sec{ex-qecc} and \sec{ft-qbf} corresponds to the initialization and encoding stages. 
The $CORRECT$ notation corresponds to the syndrome detection and error correction stages.
The $DECODE$ notation corresponds to the syndrome clean-up and decoding stages.

The program from \fig{ec-example} can be represented more succinctly by the quantum circuit shown in \fig{bit-flip-ec}.
Note that in quantum circuits each wire corresponds to a qubit, and gates are applied from left to right.
Vertical lines between wires represent multi-qubit gates.

\begin{figure}
  \mbox{
    \Qcircuit @C=1em @R=1em {
    \ket{q_1}&& \ctrl{1} & \ctrl{2} &\multigate{2}{\overline{I}}& \ctrl{3} &\qw &\qw&\qw  &\qw&\qw&\targ&\qw&\qw  & \ctrl{2}& \ctrl{1}&\qw\\
    \ket{0}&& \targ& \qw &\ghost{I}& \qw& \ctrl{2}&\ctrl{3}&\qw &\targ&\qw&\qw&\qw &\qw& \qw & \targ&\qw\\
    \ket{0}&& \qw & \targ & \ghost{I}&\qw & \qw &\qw&\ctrl{2}&\qw&\targ&\qw&\qw&\qw & \targ&\qw&\qw \\
    \ket{0}&& \qw & \qw &\qw &\targ &\targ&\qw&\qw &\ctrl{-2}&\ctrlo{-1}&\ctrl{-3}&\qw&\measuretab{q_4:=\ket{0}}&&&\\
    \ket{0}&& \qw & \qw & \qw&\qw& \qw &\targ&\targ&\ctrl{-3}&\ctrl{-1}&\ctrlo{-4}&\qw &\measuretab{q_5:=\ket{0}}&&&\
    \gategroup{1}{5}{3}{5}{.7em}{--}
    }
  }
  \caption{Identity operation protected by a quantum repetition code that can correct one bit flip ($X$) error. This is program $\widetilde{P_2}$ from \sec{ex-qecc}.}
  \label{fig:bit-flip-ec}
\end{figure}

The repetition code that corrects phase flip errors instead of bit flip errors is similar to the code already discussed, so we will not describe it in detail.
We only present the relevant circuit (\fig{phase-flip-ec}). 
We refer interested readers to \citet[Section 10.1]{MI2002}.

\begin{figure}
  \mbox{
    \Qcircuit @C=1em @R=1em {
    \ket{q_1}&& \ctrl{1} & \ctrl{2} &\gate{H}&\multigate{2}{I}&\gate{H}& \ctrl{3} &\qw &\qw&\qw  &\qw&\qw&\targ&\qw&\qw  & \ctrl{2}& \ctrl{1}&\qw\\
    \ket{0}&& \targ& \qw &\gate{H}&\ghost{I}&\gate{H}& \qw& \ctrl{2}&\ctrl{3}&\qw &\targ&\qw&\qw&\qw &\qw& \qw & \targ&\qw\\
    \ket{0}&& \qw & \targ &\gate{H}& \ghost{I}&\gate{H}&\qw & \qw &\qw&\ctrl{2}&\qw&\targ&\qw&\qw&\qw & \targ&\qw&\qw \\
    \ket{0}&& \qw & \qw&\qw &\qw &\qw&\targ &\targ&\qw&\qw &\ctrl{-2}&\ctrlo{-1}&\ctrl{-3}&\qw&\measuretab{q_4:=\ket{0}}&&&\\
    \ket{0}&& \qw & \qw&\qw & \qw&\qw&\qw& \qw &\targ&\targ&\ctrl{-3}&\ctrl{-1}&\ctrlo{-4}&\qw &\measuretab{q_5:=\ket{0}}&&&\
    \gategroup{1}{6}{3}{6}{.7em}{--}
    }
  }
  \caption{Identity operation protected by a quantum repetition code that can correct one phase flip ($Z$) error. This is program $\widetilde{P_3}$ from \sec{ex-qecc}.}
  \label{fig:phase-flip-ec}
\end{figure}
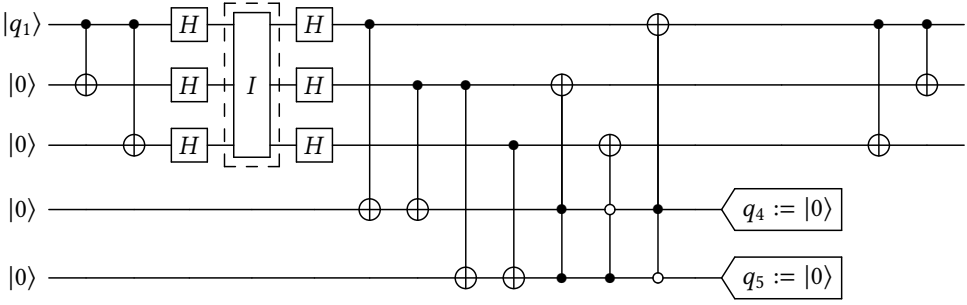

\subsection{Computation Details for the Error Correction Example}\label{app:ec-example}

In this section, we show how to compute the values for $\nm{\sem{P_1}-\sem{\widetilde{P_1}}}_{\Diamond}$, $\nm{\sem{P_2}-\sem{\widetilde{P_2}}}_{\Diamond}$, and $\nm{\sem{P_3}-\sem{\widetilde{P_3}}}_{\Diamond}$ presented in \sec{ex-qecc}.
To begin, observe that 
\begin{align}
\nm{\sem{P_1}-\sem{\widetilde{P_1}}}_{\Diamond} &= \nm{((1-p)I\circ I + pX\circ X) - I\circ I}_{\Diamond}\\
& = p\nm{I\circ I - X\circ X}_{\Diamond}\\
& = p\label{eq:IXperfdisting}
\end{align}
where $\eq{IXperfdisting}$ follows from the fact that $I$ and $X$ are perfectly distinguishable (consider $\nm{X\ket{0}-I\ket{0}}_{1} = 1$). 

Now we consider $\nm{\sem{P_2}-\sem{\widetilde{P_2}}}_{\Diamond}$. 
For program $\widetilde{P_2}$, the denotation of the noisy version of $\overline{q} := \overline{I}[\overline{q}]$ is given by
\begin{align}
& ((1-p) I\circ I + pX_1\circ X_1)\otimes ((1-p) I\circ I + pX_2\circ X_2) \otimes ((1-p) I\circ I + pX_3\circ X_3)\\
& \qquad = (1-p)^3 I\circ I + \sum_{i=1}^3 p(1-p)^2 X_i\circ X_i \nonumber\\ 
& \qquad\quad + \sum_{i,j\in \{1,2,3\}, i < j} p^2(1-p) X_iX_j\circ X_iX_j + p^3 X_1X_2X_3\circ X_1X_2X_3
\end{align} 
where $X_i$ denotes an $X$ error on qubit $q_i$. After error correction, this channel becomes
\[ ((1-p)^3 + 3p(1-p)^2) I\circ I + ( 3p^2(1-p)^2 + p^3)X_1X_2X_3\circ X_1X_2X_3. \]
The reason that the final channel is not the identity is that the three-qubit bit flip repetition code can only correct one $X$ error.
If zero or one $X$ errors occurs, then error correction will perfectly restore the state.
However, if more than one $X$ error occurs, error correction will ``correct'' the state $\rho$ to be $X_1X_2X_3\rho X_1X_2X_3$.

Recall from \app{rep-codes} that the state before applying the identity operation is $\alpha\ket{000} + \beta\ket{111}$. 
Hence, after correction, with probability $(1-p)^3 + 3p(1-p)^2$ the state remains unchanged, while with probability $3p^2(1-p) + p^3$ the state evolves to $\alpha\ket{111} + \beta\ket{000}$. 
After decoding, $q_2$ and $q_3$ are reset to $\ket{0}$ in either case.
This means that
\begin{align}
\sem{\widetilde{P_2}} &= ((1-p)^3 + 3p(1-p)^2) I\circ I + (3p^2(1-p) + p^3) X_1\circ X_1. \label{eq:P2-sem}
\end{align} 
It follows that 
\begin{align}
\nm{\sem{P_2}-\sem{\widetilde{P_2}}}_\Diamond &= \nm{( ((1-p)^3 + 3p(1-p)^2) I\circ I + (3p^2(1-p) + p^3) X_1\circ X_1)-I\circ I}_{\Diamond}\\
& = (3p^2(1-p) + p^3)(\nm{I\circ I - X\circ X}_{\Diamond})\\
& = 3p^2(1-p) + p^3.
\end{align}
This shows that $\nm{\sem{P_1}-\sem{\widetilde{P_1}}}_{\Diamond} > \nm{\sem{P_2}-\sem{\widetilde{P_2}}}_{\Diamond}$ for $p < \frac{1}{2}$.

Now we consider $\nm{\sem{P_3}-\sem{\widetilde{P_3}}}_{\Diamond}$. 
In program $\widetilde{P_3}$, the fault-tolerant identity operation is preceded and followed by $H$ gates.
Because $HXH = Z$, the $X$ errors in the previous equations are transformed into $Z$ errors.
So the denotation of the noisy version of $\overline{q} := \overline{I}[\overline{q}]$ in program $\widetilde{P_3}$ is given by
\begin{align}
  & ((1-p) I\circ I + pZ_1\circ Z_1)\otimes ((1-p) I\circ I + pZ_2\circ Z_2) \otimes ((1-p) I\circ I + pZ_3\circ Z_3)\\
  & \qquad = (1-p)^3 I\circ I + \sum_{i=1}^3 p(1-p)^2 Z_i\circ Z_i \\ 
  & \qquad \qquad\qquad+ \sum_{i,j\in \{1,2,3\}, i < j} p^2(1-p) Z_iZ_j\circ Z_iZ_j + p^3 Z_1Z_2Z_3\circ Z_1Z_2Z_3
\end{align} 
where $Z_i$ denotes a $Z$ error on qubit $q_i$. After error correction, this channel becomes
\[ ((1-p)^3 + 3p^2(1-p)) I\circ I + (3p(1-p)^2 + p^3) Z_1\circ Z_1 \]
Similar to the previous analysis, the state $\alpha\ket{000} + \beta\ket{111}$ will remain unchanged in the presence of an even number of $Z$ errors, so correction will leave the system in the correct state.
However, the state $\alpha\ket{000} + \beta\ket{111}$ will become $\alpha\ket{000} - \beta\ket{111}$, which is equivalent to $Z_1(\alpha\ket{000} + \beta\ket{111})$, in the presence of an odd number of $Z$ errors. 
The bit flip repetition code cannot correct (or detect) this type of error, so it will leave the state as-is.

Therefore, we have that 
\begin{align}
\sem{\widetilde{P_3}} &= ((1-p)^3 + 3p^2(1-p)) I\circ I + (3p(1-p)^2 + p^3) Z_1\circ Z_1. \label{eq:P3-sem}
\end{align} 
It follows that 
\begin{align}
\nm{\sem{P_3}-\sem{\widetilde{P_3}}}_\Diamond &= \nm{( ((1-p)^3 + 3p^2(1-p)) I\circ I + (3p(1-p)^2 + p^3) Z_1\circ Z_1)-I\circ I}_{\Diamond}\\
& = (3p(1-p) + p^3)(\nm{I\circ I - Z\circ Z}_{\Diamond})\\
& = 3p(1-p)^2 + p^3
\end{align}
This shows that $\nm{\sem{P_1}-\sem{\widetilde{P_1}}}_{\Diamond} < \nm{\sem{P_3}-\sem{\widetilde{P_3}}}_{\Diamond}$ for $p < \frac{1}{2}$.

% \else
% \fi

\end{document}